\RequirePackage{amsmath}
\documentclass[runningheads]{llncs}

\usepackage[T1]{fontenc}
\usepackage{amssymb}
\usepackage{lmodern} \normalfont %
\DeclareFontShape{T1}{lmr}{bx}{sc} { <-> ssub * cmr/bx/sc }{}
\usepackage{graphicx}
\usepackage{cite}
\usepackage{xfrac}
\usepackage{orcidlink}
\usepackage{todonotes}
\usepackage[ruled,linesnumbered,noend]{algorithm2e}
\usepackage{booktabs}
\usepackage{rotating}
\usepackage{enumitem}
\usepackage[symbol,para]{footmisc}

\usepackage{tikz}
\usetikzlibrary{automata, arrows.meta, shapes.geometric, calc, fit}
\tikzset{->, >=Stealth, every state/.style={thick}}
\newcommand*{\tikzmk}[1]{\tikz[remember picture,overlay,] \node (#1) {};\ignorespaces}
\newcommand{\boxit}[1]{\tikz[remember picture,overlay]{\node[xshift=-38.5pt,yshift=2.5pt,fill=#1,opacity=.2,fit={(A)($(B)+(.97\linewidth,.8\baselineskip)$)}] {};}\ignorespaces}
\newcommand{\boxitt}[1]{\tikz[remember picture,overlay]{\node[xshift=4pt,yshift=2.5pt,fill=#1,opacity=.2,fit={(A)($(B)+(-5pt,0)$)}] {};}\ignorespaces}

\usepackage{array}
\newcommand{\PreserveBackslash}[1]{\let\temp=\\#1\let\\=\temp}
\newcolumntype{C}[1]{>{\PreserveBackslash\centering}p{#1}}
\newcolumntype{R}[1]{>{\PreserveBackslash\raggedleft}p{#1}}
\newcolumntype{L}[1]{>{\PreserveBackslash\raggedright}p{#1}}

\usepackage[amsmath]{ntheorem}
\usepackage{cleveref}

\theoremnumbering{arabic}
\theoremseparator{.}
\theorembodyfont{\itshape}
\newtheorem{mytheorem}{Theorem}[section]
\newtheorem{mylemma}[mytheorem]{Lemma}
\newtheorem{myproposition}[mytheorem]{Proposition}

\theorembodyfont{\rmfamily}
\newtheorem{mynotation}[mytheorem]{Notation}
\newtheorem{myremark}[mytheorem]{Remark}
\newtheorem{myexample}[mytheorem]{Example}
\newtheorem{mydefinition}[mytheorem]{Definition}

\newtheorem{myproblem}[mytheorem]{Problem}

\creflabelformat{enumi}{#2(#1)#3}
\crefname{mytheorem}{Thm.}{Thms}
\crefname{mydefinition}{Def.}{Defs}
\crefname{myproposition}{Prop.}{Props}
\crefname{myremark}{Rem.}{Remarks}
\crefname{mylemma}{Lem.}{Lemmas}
\crefname{myexample}{Example}{Examples}
\crefname{myproof}{Proof.}{Proofs}
\crefname{myproblem}{Prob.}{Problems}
\crefname{myexample}{Ex.}{Exs}
\crefname{appendix}{Appendix}{Appendixes}
\crefname{algorithm}{Alg.}{Algs}
\crefformat{section}{{\S}#2#1#3}
\crefname{figure}{Fig.}{Figs}
\Crefname{equation}{}{}
\crefname{table}{Table}{Tables}
\crefname{line}{Line}{Lines}

\usepackage{comment}
\specialcomment{auxproof}
{\mbox{}\newline\textbf{BEGIN: AUX-PROOF}\dotfill\newline}
{\mbox{}\newline\textbf{END: AUX-PROOF}\dotfill\newline}
\excludecomment{auxproof}

\newcommand{\myparagraph}[1]{\smallskip\noindent \textbf{#1}}

\newcommand{\myparagraphii}[2]{\smallskip\textit{#1} #2}
\newcommand{\sg}{{\mathcal{G}}}
\newcommand{\MaxState}{{S_\vartriangle}}
\newcommand{\MinState}{{S_\triangledown}}
\newcommand{\MaxStateX}{{S'_\vartriangle}}
\newcommand{\MinStateX}{{S'_\triangledown}}
\newcommand{\init}{{s_\mathrm{init}}}
\newcommand{\dist}{\textup{\textsc{Dist}}}
\newcommand{\MaxStrategy}{{\sigma_\vartriangle}}
\newcommand{\MinStrategy}{{\sigma_\triangledown}}
\newcommand{\bellman}{{\mathrm{\Phi}}}
\newcommand{\wpg}{{\mathrm{\Omega}}}
\newcommand{\wpbellman}{{\mathrm{\Psi}}}
\newcommand{\gm}{{\mathrm{\Gamma}}}
\newcommand{\mcbellman}{\bellman}
\newcommand{\mcbellmanMod}{\bellman'}
\newcommand{\mcvaluefunc}{V}
\newcommand{\op}{\mathrm{op}}

\begin{document}
\title{Widest Path Games and Maximality Inheritance in Bounded Value Iteration for Stochastic Games}
\titlerunning{Widest Path Games and Maximality Inheritance in BVI for SGs}
\author{
  Kittiphon Phalakarn\inst{1}\orcidlink{0009-0006-5406-7480}%
  \and%
  Yun Chen Tsai\inst{1,2}\orcidlink{0009-0003-7705-9609}%
  \and%
  Ichiro Hasuo\inst{1,2,3}\orcidlink{0000-0002-8300-4650}%
}
\authorrunning{K.\ Phalakarn et al.}
%
\institute{
  National Institute of Informatics, Tokyo, Japan\\
  \email{\{kphalakarn,yctsai,hasuo\}@nii.ac.jp}
  \and%
  SOKENDAI (The Graduate University for Advanced Studies), Kanagawa, Japan%
  \and
  Imiron Co., Ltd., Tokyo, Japan
}
\maketitle              
\begin{abstract}
For model checking stochastic games (SGs), \emph{bounded value iteration (BVI)} algorithms have gained attention as efficient approximate methods with rigorous precision guarantees. However, BVI may not terminate or converge when the target SG contains end components. Most existing approaches address this issue by explicitly detecting and processing end components---a process that is often computationally expensive. An exception is the \emph{widest path-based BVI} approach previously studied by Phalakarn et al., which we refer to as \emph{1WP-BVI}. The method performs particularly well in the presence of numerous end components. Nonetheless, its theoretical foundations remain somewhat ad hoc. In this paper, we identify and formalize the core principles underlying the widest path-based BVI approach by (i) presenting \emph{2WP-BVI}, a clean BVI algorithm based on \emph{(2-player) widest path games}, and (ii) proving its correctness using what we call the \emph{maximality inheritance principle}---a proof principle previously employed in a well-known result in probabilistic model checking. Our experimental results demonstrate the practical relevance and potential of our proposed 2WP-BVI algorithm.
\end{abstract}
\section{Introduction}\label{sec:intro}

\myparagraph{Stochastic Games.}
Stochastic games (SGs) are two-player games played on probabilistic transition systems (also called 2.5-player games). On each player's turn, the player chooses an \emph{action} available in the current state, aiming to achieve their \emph{objectives}. The chosen action induces a transition to a successor state, determined by a predefined probabilistic distribution.

When considering the \emph{reachability objective}, the two main players are Maximizer (denoted by $\vartriangle$) and Minimizer (denoted by $\triangledown$). The Maximizer's goal is to maximize the probability of reaching a designated target set from the current state, while the Minimizer's objective is to minimize this probability.

Stochastic games have been studied across various areas of computer science. On the theoretical side, several problems---such as solving games with parity and mean-payoff objectives---have been reduced to SGs~\cite{DBLP:conf/isaac/AnderssonM09,DBLP:journals/corr/abs-1106-1232}. While the problem of solving SGs under the reachability objective remains without a known polynomial-time algorithm, its complexity class is known to be \textbf{UP} $\cap$ \textbf{coUP}~\cite{hoffman1966nonterminating}.

On the practical side, SGs are used to model and analyze control systems under probabilistic uncertainties. These include cyber-physical systems, networked systems, probabilistic programs, and applications in computer security~\cite{DBLP:journals/ejcon/SvorenovaK16}.

\myparagraph{Bounded Value Iteration.} A central goal in the study of SGs is to compute the \emph{optimal} reachability probability---that is, the probability achieved when both Maximizer and Minimizer always take optimal actions. Several techniques exist for this purpose, including quadratic programming (QP), strategy iteration (SI), and value iteration (VI). In practice, when approximate results are acceptable, VI is often preferred due to its performance~\cite{DBLP:conf/tacas/HartmannsJQW23}.

VI follows a simple fixed-point principle, i.e., Kleene's theorem (\cref{thm:kleene}). In brief, VI starts with a function that assigns zero to all states and repeatedly applies the \emph{(reachability) Bellman operator} (\cref{def:reachabilityBellmanOpr}). This produces a sequence of functions---called the \emph{under-approximation sequence}---that converges to the (optimal) reachability probability from below.

However, a drawback of standard VI is that it offers no indication of how close the approximation is to the true value~\cite{DBLP:journals/tcs/HaddadM18}, and therefore provides no guarantee of precision. To address this, \emph{bounded} VI (BVI) has been introduced, where an \emph{over-approximation sequence} is constructed alongside the under-approximation. When the two sequences converge within $2\varepsilon$, their average differs from the true value by at most $\varepsilon$.

Several BVI approaches have been proposed (see related work in \cref{sec:relatedWork}). The main challenge, however, is that the two sequences may fail to converge due to a theoretical gap between the least and greatest fixed points of the Bellman operator. To ensure convergence, it is often necessary to enforce uniqueness of the fixed point---typically by handling \emph{end components}, i.e., sets of states where players can remain indefinitely.

\myparagraph{The Previous Widest Path Approach 1WP-BVI: via MDPs.} We are interested in the \emph{widest path-based BVI approach} proposed in~\cite{DBLP:conf/cav/PhalakarnTHH20} (referred to in this paper as \emph{1WP-BVI}), which offers an alternative means of enforcing fixed-point uniqueness. This approach constructs the over-approximation sequence by computing the \emph{widest path width} of a certain weighted directed graph derived from the SG. Notably, it does not rely on end-component analysis---a particularly significant advantage in the context of SGs (as opposed to Markov decision processes (MDPs)), since end components in SGs cannot be identified solely through graph-based analysis. Indeed, the algorithm demonstrates superior performance, especially in SGs with many end components~\cite{DBLP:conf/cav/PhalakarnTHH20}.

However, both the algorithm and the correctness proof of 1WP-BVI are rather intricate. This is somewhat unsatisfying, especially given that VI is based on a simple fixed-point principle. As our goal is to establish another simple fixed-point property---i.e., uniqueness---it is natural to seek a corresponding mathematical principle that clearly explains why the widest path-based BVI works.

Moreover, although not explicitly stated, 1WP-BVI is primarily designed for MDPs rather than SGs. For SGs, the approach requires an additional mechanism referred to as \emph{player reduction} (\cref{subsec:compare}), which depends on the outcome of the under-approximation sequence. This dependency makes the over-approximation sequence non-standalone and contributes to the complication of the correctness proof. Due to the use of player reduction in~\cite{DBLP:conf/cav/PhalakarnTHH20}, which reduces the number of players from two to \emph{one}, we refer to their algorithm as 1WP-BVI (with ``1'').

\myparagraph{Identifying the Principle Behind Widest Path-Based BVI, and a New Algorithm 2WP-BVI.} Our main contribution is to identify the mathematical principle underlying the widest path-based BVI approach and to refine both the algorithm and its theoretical foundations. Specifically:
\begin{enumerate}
    \item We present a cleaner and more streamlined algorithm (called \emph{2WP-BVI} (\cref{algo_iiwpg})), based on a construct we call the \emph{(2-player) widest path game}, which avoids the use of player reduction; and
    \item For its correctness proof (\cref{subsec:ccAlg2}), we identify and apply a proof principle that we call the \emph{maximality inheritance principle}, presented in \cref{sec:maxPresPrin}.
\end{enumerate}

Regarding the first key ingredient---(2-player) widest path games (WPGs)---we find that their close structural similarity to SGs naturally leads to the use of the maximality inheritance principle. Furthermore, we show that WPGs can be solved not only using fixed-point methods (such as Kleene iteration or the Bellman--Ford algorithm), but also via an enhanced Dijkstra-like algorithm. Isolating this game structure gives 2WP-BVI theoretical advantages over 1WP-BVI.

Regarding the second key ingredient---the maximality inheritance principle---it is inspired by a classical result in probabilistic model checking~\cite[Thm.~10.19]{DBLP:books/daglib/0020348}: in a Markov chain (MC), one can enforce the uniqueness of fixed points of the Bellman operator by ``removing'' states with reachability probability $0$.

The proof of~\cite[Thm.~10.19]{DBLP:books/daglib/0020348} proceeds by contradiction, and is outlined as follows (with a detailed version provided in \cref{sec:maxPresPrin}). Assuming a gap exists between the least and the greatest fixed points, we select a \emph{gap maximizer}---a state where the gap is maximal. (The finiteness of the state space is really needed here; a separating example is in \cref{remark:finiteness}.) We then observe that \emph{maximality} (i.e., being a gap maximizer) \emph{is inherited} by successor states, leading to the conclusion that all reachable states from a gap maximizer are also gap maximizers. This yields a contradiction: a gap maximizer, by definition, cannot be a target state; thus the above argument implies that it must have reachability probability $0$; but such a state must have been removed by assumption.

In \cref{subsec:ccAlg2}, we apply this maximality inheritance principle to prove the correctness of our algorithm 2WP-BVI (\cref{algo_iiwpg}). We highlight a clear similarity between our proof and the one for~\cite[Thm.~10.19]{DBLP:books/daglib/0020348}.

\myparagraph{Contributions.} Our technical contributions are summarized as follows:
\begin{enumerate}
    \item We present 2WP-BVI, a clean BVI algorithm for stochastic games (SGs) based on the concept of widest path width. The algorithm is formalized using a construct called \emph{(2-player) widest path games}.
    \item We formulate the \emph{maximality inheritance principle} and demonstrate its application both in a classical result~\cite[Thm.~10.19]{DBLP:books/daglib/0020348} and in the correctness proof of 2WP-BVI. This provides a principled foundation that improves upon the ad hoc 1WP-BVI~\cite{DBLP:conf/cav/PhalakarnTHH20}.
    \item We evaluate 2WP-BVI via experiments (\cref{sec:exp}). As expected, it performs particularly well on SGs with many end components. Also, it significantly outperforms other tools on a class of benchmarks. This confirms the relevance and potential of 2WP-BVI.
\end{enumerate}

\myparagraph{1WP-BVI vs. 2WP-BVI.} Three main differences between 1WP-BVI and 2WP-BVI are highlighted below. A detailed comparison is provided in \cref{subsec:compare}.
\begin{enumerate}
    \item \emph{Computation of the over-approximation sequence.} 2WP-BVI constructs this sequence by solving (2-player) widest path games (\cref{algo_wpg}), whereas 1WP-BVI applies player reduction and then computes (1-player) widest path widths.
    \item \emph{Underlying principle for the convergence proof.} The fixed-point uniqueness of 2WP-BVI is by the maximality inheritance principle, but the proof of 1WP-BVI in~\cite{DBLP:conf/cav/PhalakarnTHH20} is by the infinitary pigeonhole principle and is non-constructive.
    \item \emph{Dependency between under- and over-approximations.} 2WP-BVI does not have such dependency, unlike 1WP-BVI.
\end{enumerate}

\myparagraph{Organization.} The paper is organized as follows: \cref{sec:prelim} reviews some backgrounds. The maximality inheritance principle is introduced in \cref{sec:maxPresPrin}, using~\cite[Thm.~10.19]{DBLP:books/daglib/0020348} as a showcase. In \cref{sec:WPG}, we define (2-player) widest path games and present a Dijkstra-type algorithm to solve them. Building on these key ingredients, we describe our 2WP-BVI algorithm in \cref{sec:WPGBVI} and prove its correctness and convergence. Related works are reviewed in \cref{sec:relatedWork}. The experimental results are in \cref{sec:exp}; in \cref{sec:conclusion} we conclude.


\section{Preliminaries}\label{sec:prelim}

\begin{mynotation}
For a finite set $S$, the set of all functions from $S$ to $[0,1]$ is denoted by $[0,1]^S$. The set of all discrete probability distributions on $S$ is denoted by $\dist(S) := \{d \in [0,1]^S : \sum_{s \in S} d(s) = 1\}$.
\end{mynotation}

\subsection{Stochastic Games}\label{subsec:SG}

\begin{mydefinition}[stochastic game $\sg$]\label{def:SG} A \emph{stochastic game (SG)} is $\sg = (S, \MaxState,\linebreak\MinState, \init, A, \textup{\textsc{Av}}, \delta, T)$, where $S = \MaxState \uplus \MinState$ is a finite set of \emph{states} partitioned into \emph{Maximizer's $(\MaxState)$} and \emph{Minimizer's $(\MinState)$} states, $\init \in S$ is an \emph{initial} state, $A$ is a finite set of \emph{actions}, $\textup{\textsc{Av}} : S \to 2^A \setminus \{\emptyset\}$ defines \emph{available} actions at each state, $\delta : S \times A \to \dist(S)$ is a \emph{transition function}, and $T \subseteq S$ is a \emph{target} set.
\end{mydefinition}

The semantics of SGs~\cite{DBLP:conf/cav/KelmendiKKW18} is summarized as follows. Let $\textsc{Post}(s,a) := \{s' : \delta(s,a,s') > 0\}$. An \emph{infinite path} is $\rho = s_0 a_0 s_1 a_1 \ldots \in (S \times A)^\omega$ where $a_i \in \textsc{Av}(s_i)$ and $s_{i+1} \in \textsc{Post}(s_i,a_i)$ for $i \in \mathbb{N}$. A \emph{strategy for Maximizer} is $\MaxStrategy : \MaxState \to A$ such that $\MaxStrategy(s) \in \textsc{Av}(s)$ for $s \in \MaxState$. A \emph{strategy for Minimizer} $\MinStrategy : \MinState \to A$ is similarly defined. We only consider strategies of this form (i.e., memoryless pure strategies) as they are complete for finite SGs with the reachability objective~\cite{DBLP:journals/iandc/Condon92}.

Given an SG $\sg$, a pair of strategies $(\MaxStrategy,\MinStrategy)$ induces a Markov chain $\sg^{(\MaxStrategy,\MinStrategy)}$ with a transition function $\delta^{(\MaxStrategy,\MinStrategy)} : S \to \dist(S)$ such that $\delta^{(\MaxStrategy,\MinStrategy)}(s,s') = \delta(s,\MaxStrategy(s),s')$ if $s \in \MaxState$ and $\delta(s,\MinStrategy(s),s')$ if $s \in \MinState$. The Markov chain assigns to each $s \in S$ a probability distribution $\mathcal{P}^{(\MaxStrategy,\MinStrategy)}_s$ over $S^\omega$. For each measurable subset $X \subseteq S^\omega$, $\mathcal{P}^{(\MaxStrategy,\MinStrategy)}_s(X)$ is the probability that, when starting at $s$, the Markov chain generates an infinite path which belongs to $X$.

For the \emph{reachability objective} with a target set $T \subseteq S$, we consider the measurable subset of infinite paths which contain a state in $T$, denoted by $\lozenge T$. We are interested in $\mathcal{P}^{(\MaxStrategy,\MinStrategy)}_s(\lozenge T)$ when both players play optimally, defined in the following definition as the \emph{(reachability) value function}. Informally, Maximizer aims to maximize the reachability probability to $T$ while Minimizer aims to minimize it. As proved in~\cite{DBLP:journals/iandc/Condon92}, the player order of choosing a strategy is irrelevant.

\begin{mydefinition}[(reachability) value function $V$]\label{def:reachabilityValueFunc}
Let $\sg$ be an SG as in \cref{def:SG}. The \emph{(reachability) value function} is the function $V : S \to [0,1]$ defined by $\displaystyle V(s) := \max_\MaxStrategy \min_\MinStrategy \mathcal{P}^{(\MaxStrategy,\MinStrategy)}_s(\lozenge T) = \min_\MinStrategy \max_\MaxStrategy \mathcal{P}^{(\MaxStrategy,\MinStrategy)}_s(\lozenge T)$.
\end{mydefinition}

\begin{myexample}\label{ex:SG}
\cref{fig:ex_sg} (left) shows an example of SG with $T = \{s_2\}$. Its (reachability) value function is $V : s_0 \mapsto \text{\sfrac{4}{5}} \mid s_1 \mapsto \text{\sfrac{3}{5}} \mid s_2 \mapsto 1 \mid s_3,s_4,s_5 \mapsto 0$ under a pair of optimal strategies $(\MaxStrategy,\MinStrategy)$ where $\MaxStrategy(\cdot) = \alpha$ and $\MinStrategy(\cdot) = \beta$.
\end{myexample}

\begin{figure}[t]
    \centering
    \scalebox{.8}{\begin{tikzpicture}[
    maximizer/.style={regular polygon,regular polygon sides=3, inner sep=1pt, anchor=north},
    minimizer/.style={maximizer, shape border rotate=180}
]
    \node [state, maximizer] at (0,0) (s0) {$s_0$};
    \node [state, minimizer] at (2,0) (s1) {$s_1$};
    \node [state, maximizer, accepting] at (4,0) (s2) {$s_2$};
    \node [state, minimizer] at (0,-2) (s3) {$s_3$};
    \node [state, maximizer] at (2,-2) (s4) {$s_4$};
    \node [state, minimizer] at (4,-2) (s5) {$s_5$};

    \path ($(s0.150) - (3ex,0)$) edge (s0.150);
    \path (s0.30) edge node[above, pos=0.1]{$\alpha$} node[above, pos=0.75]{\sfrac{1}{3}} coordinate[pos=0.4](s01) (s1.210);
    \path (s01) edge[out=100, in=50, looseness=2] node[above left=-1.5pt, pos=0.8]{\sfrac{1}{3}} (s0.75);
    \path (s01) edge[out=60, in=140, looseness=0.9] node[above=2pt, pos=0.93]{\sfrac{1}{3}} (s2.130);
    \path (s0.252) edge node[left, pos=0.25]{$\beta$} (s3.108);
    \path (s1.330) edge node[above, pos=0.22]{$\alpha$} (s2.147);
    \path (s1) edge node[left=1pt, pos=0.15]{$\beta$} node[right=2pt, pos=0.9]{\sfrac{1}{4}} coordinate[pos=0.5](s13) (s3.30);
    \path (s13) edge node[above, pos=0.4]{\sfrac{3}{4}} (s0);
    \path (s1) edge node[right, pos=0.15]{$\gamma$} node[right=-1pt, pos=0.7]{\sfrac{1}{5}} coordinate[pos=0.5](s14) (s4);
    \path (s14) edge node[above, pos=0.75]{\sfrac{4}{5}} (s2.210);
    \path (s2) edge[out=15, in=40, looseness=15] node[right]{$\alpha$} (s2);
    \path (s3.72) edge node[right, pos=0.2]{$\alpha$} (s0.288);
    \path (s3.310) edge node[below, pos=0.1]{$\beta$} (s4.170);
    \path (s4.130) edge node[above, pos=0.1]{$\alpha$} (s3.350);
    \path (s4.30) edge node[below, pos=0.28]{$\beta$} (s5.210);
    \path (s5) edge node[right, pos=0.2]{$\alpha$} node[right=-1pt, pos=0.7]{\sfrac{1}{2}} coordinate[pos=0.5](s52) (s2);
    \path (s52) edge node[below, pos=0.78]{\sfrac{1}{2}} (s4.78);
    \path (s5) edge[out=320, in=345, looseness=15] node[right]{$\beta$} (s5);
\end{tikzpicture}}\qquad\scalebox{.8}{\begin{tikzpicture}[
    maximizer/.style={regular polygon,regular polygon sides=3, inner sep=1pt, anchor=north},
    minimizer/.style={maximizer, shape border rotate=180}
]
    \node [state, maximizer] at (0,0) (s0) {$s_0$};
    \node [state, minimizer] at (2,0) (s1) {$s_1$};
    \node [state, maximizer, accepting] at (4,0) (s2) {$s_2$};
    \node [state, minimizer] at (0,-2) (s3) {$s_3$};
    \node [state, maximizer] at (2,-2) (s4) {$s_4$};
    \node [state, minimizer] at (4,-2) (s5) {$s_5$};

    \path ($(s0.150) - (3ex,0)$) edge (s0.150);
    \path (s0.30) edge node[above, pos=0.17]{\sfrac{1}{5}} coordinate[pos=0.4](s01) (s1.210);
    \path[-] (s0.30) edge[ultra thick] (s01);
    \path (s01) edge[out=100, in=50, looseness=2] (s0.75);
    \path (s01) edge[out=60, in=140, looseness=0.9, ultra thick] (s2.130);
    \path (s0.252) edge node[left, pos=0.25]{\sfrac{3}{5}} (s3.108);
    \path (s1.330) edge node[above, pos=0.27]{\sfrac{2}{5}} (s2.147);
    \path (s1) edge node[left=1pt, pos=0.15]{\sfrac{3}{20}} coordinate[pos=0.5](s13) (s3.30);
    \path[-] (s1) edge[ultra thick] (s13);
    \path (s13) edge[ultra thick] (s0);
    \path (s1) edge node[right, pos=0.15]{\sfrac{12}{25}} coordinate[pos=0.5](s14) (s4);
    \path (s14) edge (s2.210);
    \path (s2) edge[out=15, in=40, looseness=15] node[right]{\sfrac{2}{5}} (s2);
    \path (s3.72) edge node[right, pos=0.2]{0} (s0.288);
    \path (s3.310) edge[ultra thick] node[below, pos=0.1]{\sfrac{4}{5}} (s4.170);
    \path (s4.130) edge[ultra thick] node[above, pos=0.1]{\sfrac{3}{5}} (s3.350);
    \path (s4.30) edge node[below, pos=0.28]{1} (s5.210);
    \path (s5) edge node[right, pos=0.2]{\sfrac{3}{5}} coordinate[pos=0.5](s52) (s2);
    \path (s52) edge (s4.78);
    \path (s5) edge[out=320, in=345, looseness=15, ultra thick] node[right]{1} (s5);
\end{tikzpicture}}
    \caption{Left: an example of SG. The labels $1$ for $\delta(s,a,s') = 1$ are omitted. Right: the WPG constructed from the SG and $u : s_i \mapsto \text{\sfrac{$i$}{5}}$ for $0 \leq i \leq 5$. Each $(s,a)$ is labeled by its width $\phi_u(s,a)$. Players' optimal strategies for the widest path objective are in bold.}
    \label{fig:ex_sg}\label{fig:cwpvf}
\end{figure}
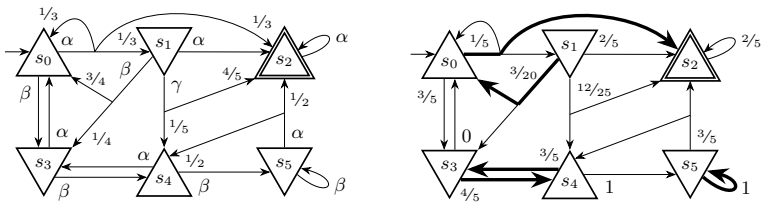

\subsection{Fixed Points in a Complete Lattice}\label{subsec:fixedPtPrelim}

We recall some backgrounds on lattices. Although their usages are not limited to probabilistic model checking, we use SGs here for intuitions and illustrations.

Our focus in this paper is the complete lattice $(L = [0,1]^S, \preccurlyeq)$ where $f \preccurlyeq f'$ iff $f(s) \leq f'(s)$ for all $s \in S$. Its least element is $\bot : S \to [0,1]$ with $\bot(s)=0$ and its greatest element is $\top : S \to [0,1]$ with $\top(s)=1$ for all $s \in S$. A function $\gm : L \to L$ is \emph{monotone} if $f \preccurlyeq f'$ implies $\gm f \preccurlyeq \gm f'$. Also, $f \in L$ is a \emph{fixed point} of $\gm : L \to L$ if $f = \gm f$. Below we state two well-known theorems on fixed points.

\begin{mydefinition}[Cousot--Cousot sequence]\label{def:CCseq}
Let $(L, \preccurlyeq)$ be a complete lattice and $\gm : L \to L$ be a monotone function. The \emph{bottom-up Cousot--Cousot sequence} and the \emph{top-down Cousot--Cousot sequence} are the (transfinite) sequences
\begin{align*}
& \bot \preccurlyeq \gm\bot \preccurlyeq \gm^2\bot \preccurlyeq \cdots \preccurlyeq \gm^\omega\bot \preccurlyeq \gm^{\omega+1}\bot \preccurlyeq \cdots \preccurlyeq \gm^\alpha\bot \preccurlyeq \cdots \quad\text{and}
\\
&
\top \succcurlyeq \gm\top \succcurlyeq \gm^2\top \succcurlyeq \cdots \succcurlyeq \gm^\omega\top \succcurlyeq \gm^{\omega+1}\top \succcurlyeq \cdots \succcurlyeq \gm^\alpha\top \succcurlyeq \cdots,
\end{align*}
where $\gm^{\alpha+1}\bot = \gm(\gm^{\alpha}\bot)$ and $\gm^{\alpha+1}\top = \gm(\gm^{\alpha}\top)$ for a successor ordinal $\alpha + 1$, $\gm^\alpha\bot = \sup\{\gm^\beta\bot : \beta < \alpha\}$ and $\gm^\alpha\top = \inf\{\gm^\beta\top : \beta < \alpha\}$ for a limit ordinal $\alpha$.
\end{mydefinition}

\begin{mytheorem}[Cousot--Cousot~\cite{DBLP:conf/popl/CousotC77}]\label{thm:cousot}
Assume the setting of \cref{def:CCseq}. The bottom-up Cousot--Cousot sequence stabilizes, i.e., there is an ordinal $\alpha_0$ such that $\gm^{\alpha_0}\bot = \gm^{\alpha_0+1}\bot = \cdots$. Its limit $\gm^{\alpha_0}\bot$ is the least fixed point (lfp) $\mu\gm$ of $\gm$.

Dually, the top-down Cousot--Cousot sequence stabilizes: $\gm^{\alpha'_0}\top = \gm^{\alpha'_0+1}\top = \cdots$ with some $\alpha'_0$. Its limit $\gm^{\alpha'_0}\top$ is the greatest fixed point (gfp) $\nu\gm$ of $\gm$.
\end{mytheorem}

\begin{mytheorem}[(special case of) Kleene~\cite{DBLP:journals/dm/Baranga91}]\label{thm:kleene}
Let $(L,\preccurlyeq)$ be a complete lattice and $\gm : L \to L$ be an $\omega$-continuous function (i.e., $\gm(\sup L') = \sup\{\gm f : f \in L'\}$ for any increasing $\omega$-chain $L' \subseteq L$). Then, the \emph{bottom-up Kleene sequence} $\bot \preccurlyeq \gm\bot \preccurlyeq \gm^{2}\bot \preccurlyeq \cdots$ stabilizes at $\omega$ $(i.e., \gm^{\omega}\bot = \gm^{\omega+1}\bot = \cdots)$. Its limit $\gm^\omega\bot$ is the lfp $\mu\gm$ of $\gm$.

Dually, let $\gm$ be an $\omega^{\op}$-cocontinuous function (i.e., $\gm(\inf L') = \inf\{\gm f : f \in L'\}$ for any decreasing $\omega$-chain $L' \subseteq L$). Then, the \emph{top-down Kleene sequence} $\top \succcurlyeq \gm\top \succcurlyeq \gm^{2}\top \succcurlyeq \cdots$ stabilizes at $\omega$, with $\gm^\omega\top = \nu\gm$ giving the gfp.
\end{mytheorem}

The bottom-up sequence is the essence behind the \emph{value iteration} technique in probabilistic model checking (\cref{sec:intro}). Thus, in this work, we generally refer to the iterative construction of such a sequence as \emph{value iteration}. Note that the bottom-up and top-down sequences are the \emph{under-} and \emph{over-approximation sequences} of $\mu\gm$, respectively (i.e., $\gm^{\alpha}\bot \preccurlyeq \mu\gm \preccurlyeq \nu\gm \preccurlyeq \gm^{\alpha'}\top$ for any ordinals $\alpha$ and $\alpha'$).

\subsection{Bounded Value Iteration}\label{subsec:BVI}

The problem we would like to solve is as follows.

\begin{myproblem}[value function approximation problem]\label{def:prob}
Given an SG $\sg$ and a precision $\varepsilon > 0$, the \emph{value function approximation problem} is to output a value that is $\varepsilon$-close to $V(\init)$ (i.e., the difference from $V(\init)$ is at most $\varepsilon$).
\end{myproblem}

\cref{def:prob} can be practically solved by approximating the lfp of the \emph{(reachability) Bellman operator}. We first introduce the notion of the \emph{state-action expectation} and then define the (reachability) Bellman operator based on this notion.

\begin{mydefinition}[state-action expectation $\phi_f(s,a)$]\label{def:stateActionExp}
Let $\sg$ be an SG as in \cref{def:SG}. The \emph{state-action expectation} of a function $f \in [0,1]^S$ for a state-action pair $(s,a) \in S \times A$ is $\phi_f(s,a) := \sum_{s' \in S} \delta(s,a,s') \cdot f(s')$.
\end{mydefinition}

\begin{mydefinition}[(reachability) Bellman operator $\bellman$]\label{def:reachabilityBellmanOpr}
Let $\sg$ be an SG as in \cref{def:SG}. The \emph{(reachability) Bellman operator} is the function $\bellman : [0,1]^S \to [0,1]^S$ where, for $f \in [0,1]^S$,
$$
\begin{array}{l}
    (\bellman f)(s) :=
    \begin{cases}
        1 & \textnormal{if } s \in T, \\
        \max_{a \in \textup{\textsc{Av}}(s)} \phi_f(s,a) & \textnormal{if } s \in \MaxState\setminus T, \textnormal{ and} \\
        \makebox[\widthof{$\max_{a \in \textup{\textsc{Av}}(s)} \phi_f(s,a)$}][r]{$\min_{a \in \textup{\textsc{Av}}(s)} \phi_f(s,a)$} & \textnormal{if } s \in \MinState\setminus T.
    \end{cases}
\end{array}
$$
\end{mydefinition}

\noindent Here is a well-known theorem on SG reachability probability (see, e.g.,~\cite{DBLP:conf/spin/ChatterjeeH08}).

\begin{mytheorem}\label{thm:wellknown}
Let $V$ and $\bellman$ be as in \cref{def:reachabilityValueFunc,def:reachabilityBellmanOpr}. We have $\mu \bellman = V$.
\end{mytheorem}

The (reachability) Bellman operator $\bellman$ is known to be $\omega$-continuous and thus its lfp can be (under-)approximated using Kleene's theorem (\cref{thm:kleene}): $\bot \preccurlyeq \bellman\bot \preccurlyeq \bellman^2\bot \preccurlyeq \cdots$ converges to $\mu\bellman = V$. Nonetheless, one cannot know how close the approximation $(\bellman^i \bot)(\init)$ is to $V(\init)$~\cite{DBLP:journals/tcs/HaddadM18}. In particular, even when $(\bellman^i \bot)(\init) - (\bellman^{i-1} \bot)(\init) \leq 2\varepsilon$---a common stopping criterion for (non-bounded) VI---the approximation $(\bellman^i \bot)(\init)$ may not be $\varepsilon$-close to $V(\init)$~\cite{DBLP:journals/tcs/HaddadM18}. To overcome this, over-approximations of $\mu\bellman$ are brought into consideration.

\emph{Bounded value iteration (BVI)} (also known as \emph{interval iteration}) computes, in addition to the under-approximation sequence, the over-approximation sequence $\top \succcurlyeq \bellman\top \succcurlyeq \bellman^2\top \succcurlyeq \cdots$ which converges to $\nu\bellman \succcurlyeq V$. If one finds $i \in \mathbb{N}$ such that $(\bellman^i \top)(\init) - (\bellman^i \bot)(\init) \leq 2\varepsilon$, then the value $\frac{1}{2}((\bellman^i \bot)(\init) + (\bellman^i \top)(\init))$ is guaranteed to be $\varepsilon$-close to $V(\init)$. The issue with BVI is that $\nu\bellman$ is not necessarily equal to $\mu\bellman = V$. This means the over-approximation sequence may not converge to $V$ and one may not be able to find $i \in \mathbb{N}$ that satisfies the condition. Several studies have been conducted to modify the over-approximation sequence in order to obtain convergence, e.g., by player reduction and end-component analysis. We refer to~\cite{DBLP:conf/cav/KelmendiKKW18,DBLP:conf/lics/KretinskyMW23,DBLP:conf/cav/PhalakarnTHH20} for details; see also the related works in \cref{sec:relatedWork}.

\begin{myexample}\label{ex:bvi}
Consider the SG in \cref{fig:ex_sg} (left). By representing $f \in [0,1]^S$ as $(f(s_0),\ldots,f(s_5))$, the under-approximation sequence is $\bot = (0,0,0,0,0,0) \preccurlyeq (0,0,1,0,0,0) \preccurlyeq (\text{\sfrac{1}{3}},0,1,0,0,0) \preccurlyeq (\text{\sfrac{4}{9}},\text{\sfrac{1}{4}},1,0,0,0) \preccurlyeq (\text{\sfrac{61}{108}},\text{\sfrac{1}{3}},1,0,0,0) \preccurlyeq \cdots$ converging to $V = (\text{\sfrac{4}{5}},\text{\sfrac{3}{5}},1,0,0,0)$. On the other hand, the over-approximation sequence is $\top = (1,1,1,1,1,1) \succcurlyeq (1,1,1,1,1,1) \succcurlyeq \cdots$ stabilizing at $\nu\bellman = \top$.
\end{myexample}

\section{The Maximality Inheritance Principle}\label{sec:maxPresPrin}

To introduce our second key ingredient, we revisit~\cite[Thm.~10.19]{DBLP:books/daglib/0020348}, a classical result well-known in the probabilistic model checking community.

\begin{mynotation}
The notations used in \cref{sec:maxPresPrin} are limited to this section (where we only consider Markov chains). For example, $\bellman$ denotes the one in \cref{eq:MCBellman} in this section, while it denotes the one in \cref{def:reachabilityBellmanOpr} (for SGs) in other sections.
\end{mynotation}

Let $\mathcal{M}=(S,\init,\delta,T)$ be a Markov chain (MC), with a finite state space $S$, an initial state $\init$, a transition kernel $\delta\colon S\to\dist(S)$, and a set $T \subseteq S$ of target states. Let $\mcbellman : [0,1]^{S} \to [0,1]^{S}$ be the Bellman operator for reachability probability, that is, for $f\in[0,1]^{S}$,
\begin{equation*}\label{eq:MCBellman}
\mcbellman(f)(s) :=
\begin{cases}
    1 & \textnormal{if } s \in T, \textnormal{ and} \\
    \sum_{s'\in S}\delta(s,s')\cdot f(s') & \text{otherwise}.
\end{cases}
\end{equation*}
It is standard that the \emph{value function} $\mcvaluefunc : S \to [0,1]$, which assigns to each $s\in S$ the reachability probability to $T$, is the lfp of $\mcbellman $ (i.e., $\mcvaluefunc=\mu \mcbellman $).

Much like our discussion in~\S\ref{subsec:fixedPtPrelim}--\ref{subsec:BVI}, it is desired that the fixed point of $\mcbellman$ is unique (i.e., $\mu\mcbellman=\nu\mcbellman$), which would give us BVI with convergence. This is not the case in general; a counterexample can be given essentially like in \cref{ex:bvi}.

Nonetheless, achieving fixed-point uniqueness for MC reachability is not difficult. Let $S_{=0} \subseteq S$ be the set of states that reach $T$ with probability $0$. We define the \emph{modified Bellman operator}\footnote[2]{In~\cite{DBLP:books/daglib/0020348}, the set $S_{=1} \subseteq S$ of states that reach $T$ with probability $1$ is also defined. However, $S_{=1}$ is not necessary for the uniqueness proof.} $\mcbellmanMod : [0,1]^S \to [0,1]^S$ where, for $f \in [0,1]^S$,
\begin{equation}\label{eq:MCBellmanMod}
\mcbellmanMod(f)(s) :=
\begin{cases}
    1 & \textnormal{if } s \in T, \\
    0 & \textnormal{if } s \in S_{=0}, \textnormal{ and} \\
    \sum_{s'\in S}\delta(s,s')\cdot f(s') & \textnormal{otherwise}.
\end{cases}
\end{equation}
We claim that $\mcbellmanMod$ has a unique fixed point. First, we state the following.

\begin{mylemma}[max.\ vs.\ average]\label{lem:maxVsAvg}
Let $S$ be a finite set, $f\colon S\to [0,1]$ be a function, and $\delta\in \dist(S)$ be a distribution over $S$. Assume that $s^\star$ is an $f$-maximizer (meaning $f(s^\star) \geq f(s)$ for each $s\in S$).
\begin{enumerate}[label=(\roman*)]
    \item $f(s^\star)\geq \sum_{s\in S}\delta(s)\cdot f(s)$ holds (``the maximum is above the average'').
    \item Assume further that we have an equality $f(s^\star) = \sum_{s\in S}\delta(s)\cdot f(s)$ (``the maximum coincides with the average''). Then, for each $s\in S$, $\delta(s)>0$ implies $f(s^\star) = f(s)$.
\end{enumerate}
\end{mylemma}

\noindent We provide a uniqueness proof as follows.

\begin{myproposition}[unique solution {\cite[Thm.~10.19]{DBLP:books/daglib/0020348}}]\label{prop:maxPrevMC}
Assume the setting of MC in \cref{sec:maxPresPrin}. Then, the modified Bellman operator $\mcbellmanMod$ has a unique fixed point. Moreover, $\mcvaluefunc$ is its fixed point; therefore we have $\mcvaluefunc = \mu \mcbellmanMod = \nu \mcbellmanMod$.
\end{myproposition}
\begin{proof}
We first show that $\mcvaluefunc = \mu\mcbellmanMod$. This follows the fact that the Kleene sequence $\bot \preccurlyeq \mcbellmanMod(\bot) \preccurlyeq (\mcbellmanMod)^2(\bot) \preccurlyeq \cdots$
for the modified Bellman operator $\mcbellmanMod$ coincides with that for $\mcbellman$. We show this by contradiction. Pick the smallest $k$ such that $(\mcbellman)^k(\bot) \neq (\mcbellmanMod)^k(\bot)$. Since the only difference is in the second case of \cref{eq:MCBellmanMod}, there must be some $s \in S_{=0}$ such that $(\mcbellman)^k(\bot)(s) > 0$. But then we have $\mcvaluefunc(s) \geq (\mcbellman)^k(\bot)(s)>0$, contradicting with $s\in S_{=0}$.

It follows that $\mcvaluefunc \preccurlyeq \nu \mcbellmanMod$. We show that this inequality is an equality. We argue by contradiction; assume $\mcvaluefunc \prec \nu\mcbellmanMod$. Then, there is $s \in S$ such that $\mcvaluefunc(s) < \nu\mcbellmanMod(s)$. Let $s^\star$ be a \emph{gap maximizer}, i.e., a state that maximizes this gap: $ (\nu\mcbellmanMod - \mcvaluefunc)(s^\star) \geq (\nu\mcbellmanMod - \mcvaluefunc)(s)$ for each $ s\in S$. Note that we can pick such $s^\star$ since $S$ is finite.

By this choice of $s^\star$, we have $(\nu\mcbellmanMod - \mcvaluefunc)(s^\star) > 0$. It follows that
\begin{equation}\label{eq:MCSStarNotZeroOrTerm}
s^\star\not\in T\quad\text{and}\quad s^\star\not\in S_{=0}.
\end{equation}
Indeed, if we assume otherwise, we have
\begin{align*}
& (\nu\mcbellmanMod - \mcvaluefunc)(s^{\star})\\
& \;=\; \bigl(\mcbellmanMod(\nu\mcbellmanMod) - \mcbellmanMod(\mcvaluefunc)\bigr)(s^{\star}) \quad\text{since $\nu \mcbellmanMod$ and $\mcvaluefunc=\mu \mcbellmanMod$ are fixed points of $\mcbellmanMod$}\\
& \;=\; \bigl(\mcbellmanMod(\nu\mcbellmanMod)\bigr)(s^{\star}) - \bigl(\mcbellmanMod(\mcvaluefunc)\bigr)(s^{\star})\\
& \;=\; 1-1 \quad\text{(if $s^{\star}\in T$)}\quad\text{or}\quad 0-0 \quad\text{(if $s^{\star}\in S_{=0}$)} \qquad\qquad\quad\!\text{by \cref{eq:MCBellmanMod}}\\
& \;=\; 0,\qquad\text{and this contradicts with $(\nu \mcbellmanMod - \mcvaluefunc)(s^{\star}) > 0$.}
\end{align*}
Now, we reason as follows.
\begin{equation*}\label{eq:maximalityProofMC}
\begin{aligned}
& (\nu \mcbellmanMod - \mcvaluefunc)(s^{\star})\quad =\; \bigl(\mcbellmanMod(\nu \mcbellmanMod) - \mcbellmanMod(\mcvaluefunc)\bigr)(s^{\star})\\
& \;=\; \bigl(\mcbellmanMod(\nu \mcbellmanMod)\bigr)(s^{\star}) - \bigl(\mcbellmanMod(\mcvaluefunc)\bigr)(s^{\star})\\
& \;=\; \textstyle \sum_{s'\in S}\delta(s^{\star},s')\cdot (\nu\mcbellmanMod)(s') - \sum_{s'\in S}\delta(s^{\star},s')\cdot (\mcvaluefunc)(s')\qquad\text{by \cref{eq:MCSStarNotZeroOrTerm,eq:MCBellmanMod}}\\
& \;=\; \textstyle \sum_{s'\in S}\delta(s^{\star},s')\cdot (\nu \mcbellmanMod- \mcvaluefunc)(s').
\end{aligned}
\end{equation*}

We have proved that the average of $\nu\mcbellmanMod-\mcvaluefunc$ on the distribution $\delta(s^\star)$ coincides with its maximum taken at $s^\star$. Therefore, by \cref{lem:maxVsAvg}, we have $(\nu \mcbellmanMod - \mcvaluefunc)(s^\star) = (\nu \mcbellmanMod - \mcvaluefunc)(s')$ for each successor $s'$ of $s^\star$, where being a successor means $\delta(s^\star,s')>0$. Hence, we have shown that being a gap maximizer is \emph{inherited} to successors.

The above reasoning is valid for any gap maximizer $s^\star$. Thus, all offsprings of $s^\star$ (i.e., those reachable from $s^\star$) are gap maximizers. However, in \cref{eq:MCBellmanMod} we showed that no gap maximizer is in $T$ or in $S_{=0}$. Therefore, $T$ is unreachable from $s^\star$. This implies $V(s^\star)=0$ and thus $s^\star\in S_{=0}$, which is a contradiction with \cref{eq:MCSStarNotZeroOrTerm}. \qed
\end{proof}

We would like to signify a characteristic argument in the proof, calling it as the \emph{maximality inheritance principle}. It has the following features.
\begin{itemize}
    \item The goal is to show that two fixed points on $[0,1]^{S}$---the lfp for one operator $\bellman_1$ (here $\bellman$) and the gfp for another operator $\bellman_2$ (here $\bellman'$)---coincide.
    \item One argues by contradiction. Assuming $S$ is finite (this is necessary---see \cref{remark:finiteness}), we pick a \emph{gap maximizer} $s^\star$.
    \item Then, one uses \cref{lem:maxVsAvg} (especially Item (ii)) to show that \emph{maximality} (i.e., being a gap maximizer) \emph{is inherited}. This proof method seems to work in many probabilistic settings beyond MCs; indeed we use it for SGs in \cref{subsec:ccAlg2}.
    \item Finally, one derives a contradiction from this maximality inheritance. This is typically done by exploiting some ``lfp structure'' of $\mu \bellman_{1}$ or $\nu\bellman_2$. In \cref{prop:maxPrevMC}, we used the former. Our proof in \cref{subsec:ccAlg2} will use the latter: while the fixed point $\nu\bellman_2$ itself is the gfp, the operator $\bellman_2$ ($\wpg$ in \cref{subsec:ccAlg2}) is defined with suitable lfps.
\end{itemize}

\section{Widest Path Games}\label{sec:WPG}

The maximality inheritance principle (\cref{sec:maxPresPrin}) will be central to the correctness proof of our 2WP-BVI algorithm (\cref{sec:WPGBVI}). Here we introduce the other key ingredient, the \emph{(2-player) widest path games}, which are used for formalizing the algorithm.

We define a non-stochastic (2-player) widest path game (WPG) from an SG $\sg$ and a function $u \in [0,1]^S$. The game has a structure that mirrors that of $\sg$: it has the same sets of states for Maximizer and Minimizer; these two players compete against each other under what we call the \emph{widest path objective}.

\subsection{Widest Path Objective}\label{subsec:WPObj}

Given an SG $\sg$ as in \cref{def:SG}, and a function $u \in [0,1]^S$, we associate each state-action pair $(s,a) \in S \times A$ with a \emph{width} $\phi_u(s,a)$, where $\phi_{u}$ is the state-action expectation of $u$ (\cref{def:stateActionExp}). Given an infinite path $\rho$, its \emph{path width} with respect to $u$ is the bottleneck width of the path. It is defined recursively as follows.

\begin{mydefinition}[path width $w_u(\rho)$, widest path width $ \mathcal{W}^{(\MaxStrategy,\MinStrategy)}_{u,s} $]\label{def:pathWidth}
Let $\sg$ be an SG as in \cref{def:SG}, $\rho = s_0 a_0 s_1 a_1 \ldots \in (S \times A)^\omega$ be an infinite path (cf.\ \cref{subsec:SG}), and $u \in [0,1]^S$. The \emph{path width} $w_u(\rho)$ of $\rho$ with respect to $u$ is defined by
$$
w_u(\rho) :=
\begin{cases}
    1 & \textnormal{if } s_0 \in T, \\
    0 & \textnormal{if } s_i \not\in T \text{ for all } i \in \mathbb{N}, \text{ and} \\
    \displaystyle\min\bigl(\,\phi_u(s_0,a_0),\,w_u(s_1 a_1 s_2 a_2 \ldots)\,\bigr) & \textnormal{otherwise}.
\end{cases}
$$
Note that only paths that reach $T$ can have non-zero path width.

We extend the above to finite paths $\rho'=s_0 a_0 s_1 a_1 \ldots s_k \in ((S \setminus T) \times A)^* \times T$ with $s_k\in T$. In explicit terms, we have $w_u(\rho') = \min_{0 \leq i < k} \phi_u(s_i,a_i)$. In fact, we can easily see that $w_{u}(\rho')=w_{u}(\rho'\rho'')$ for any infinite path $\rho''$ from $s_k$.

Given a pair of strategies $(\MaxStrategy,\MinStrategy)$ and a state $s \in S$, we define $\textsc{Path}^{(\MaxStrategy,\MinStrategy)}_s := \{ s_0 a_0 s_1 a_1 \ldots \}$ to be the set of all infinite paths such that $s_0 = s$, $a_i = \MaxStrategy(s_i)$ if $s_i \in \MaxState$, and $a_i = \MinStrategy(s_i)$ if $s_i \in \MinState$, for each $i \in \mathbb{N}$. Now we define
\begin{equation}\label{eq:widestPath}
\mathcal{W}^{(\MaxStrategy,\MinStrategy)}_{u,s} \;:=\; \max\{w_u(\rho) : \rho \in \textsc{Path}^{(\MaxStrategy,\MinStrategy)}_s\};
\end{equation}
this is the \emph{widest path width} from $s$ following $(\MaxStrategy,\MinStrategy)$ with respect to $u$.
\end{mydefinition}

The maximum in \cref{eq:widestPath} is well-defined: although $\textsc{Path}^{(\MaxStrategy,\MinStrategy)}_s$ can be an infinite set, the maximum is asking for the widest path width from $s$ to $T$ in a suitably defined \emph{finite} graph. It is well-known that this widest path width is well-defined, and efficient algorithms such as Dijkstra's algorithm~\cite{DBLP:journals/nm/Dijkstra59} exploit this fact.

Now we define the \emph{widest path objective}. Given a function $u \in [0,1]^S$, we are interested in the widest path width $\mathcal{W}^{(\MaxStrategy,\MinStrategy)}_{u,s}$ when both players play optimally under the widest path objective. This is given by the \emph{widest path value function} defined below.

\begin{mydefinition}[widest path value function $W_u$]\label{def:widestPathValueFunc}
Assume the setting of \linebreak \cref{def:pathWidth}. The \emph{widest path value function} with respect to $u \in [0,1]^S$ is the function $W_u \colon S\to[0,1]$ defined by $W_u(s) := \displaystyle\min_\MinStrategy \max_\MaxStrategy \mathcal{W}^{(\MaxStrategy,\MinStrategy)}_{u,s}$.
\end{mydefinition}

Note that a pair of strategies $(\MaxStrategy,\MinStrategy)$ does not determine a successor state $s_{i+1} \in \textsc{Post}(s_i,a_i)$. Nevertheless, the widest path width $\mathcal{W}^{(\MaxStrategy,\MinStrategy)}_{u,s}$ in \cref{eq:widestPath} is defined as the maximum path width over all possible paths. Consequently, the widest path game can be interpreted as a non-stochastic game in which, at each state $s$, either Maximizer or Minimizer chooses an action $a \in \textsc{Av}(s)$, and Maximizer subsequently chooses a successor state from $\textsc{Post}(s,a)$.

Informally, by having Maximizer choose successor states in the widest path game, if $u$ is an over-approximation of the optimal reachability probability, then so is $W_u$. This observation forms the main idea behind our 2WP-BVI algorithm presented in \cref{sec:WPGBVI}.

\begin{myexample}\label{ex:wpg_def}
Consider the non-stochastic (2-player) widest path game constructed from SG in \cref{fig:ex_sg} (left) and $u : s_i \mapsto \text{\sfrac{$i$}{5}}$ for $0 \leq i \leq 5$. The width $\phi_u(s,a)$ of each state-action pair is shown in \cref{fig:cwpvf} (right). Its widest path value function is $W_u = (\text{\sfrac{1}{5}},\text{\sfrac{3}{20}},1,0,0,0)$ under a pair of optimal strategies $(\MaxStrategy,\MinStrategy)$ where $\MaxStrategy(\cdot) = \alpha$ and $\MinStrategy(\cdot) = \beta$. Notice that the widest paths require the successors of $s_0$ and $s_1$ under the mentioned pair of optimal strategies to be $s_2$ and $s_0$, respectively. These choices can be interpreted as being made by Maximizer.
\end{myexample}

\subsection{Widest Path Value Function as a Fixed Point}\label{subsec:WPVFasFixP}

The widest path value function can be computed by backward iteration from the target set using a suitable ``Bellman operator''. This is the intuition of Bellman--Ford's algorithm~\cite{bellman1958routing,ford1956network}. We formalize this with lattice-theoretic fixed points.

\begin{mydefinition}[widest path Bellman operator $\wpbellman_{u}$]\label{def:wpbellman}
Assume the setting \linebreak of \cref{def:pathWidth}. The \emph{widest path Bellman operator} with respect to $u \in [0,1]^S$ is the function $\wpbellman_u : [0,1]^S \to [0,1]^S$ defined as follows. For $f \in [0,1]^S$ and $s\in S$,
$$
(\wpbellman_u f)(s) :=
\begin{cases}
    1 & \textnormal{if } s \in T, \\
    \displaystyle\max_{a \in \textup{\textsc{Av}}(s)} \min\left(\phi_u(s,a), \max_{s' \in \textup{\textsc{Post}}(s,a)} f(s')\right) & \textnormal{if } s \in \MaxState\setminus T, \text{ and}\\
    \displaystyle\min_{a \in \textup{\textsc{Av}}(s)} \min\left(\phi_u(s,a), \max_{s' \in \textup{\textsc{Post}}(s,a)} f(s')\right) & \textnormal{if } s \in \MinState\setminus T.
\end{cases}
$$
\end{mydefinition}

\noindent Notice that $\wpbellman_u$ incorporates $\phi_u(s,a)$ (\cref{def:stateActionExp}) in a way similar to the reachability Bellman operator $\bellman$ (\cref{def:reachabilityBellmanOpr}). Therefore, roughly speaking, an application of $\wpbellman_u$ subsumes an application of $\bellman$. This intuition will be useful later.

\begin{mylemma}\label{lem:wpbellmanMonotone}\label{lem:wpvf}
(i) $\wpbellman_u$ is monotone, and (ii) $\mu\wpbellman_u = W_u$ (\cref{def:widestPathValueFunc}).
\end{mylemma}

\cref{lem:wpvf} suggests that one can compute $W_u = \mu\wpbellman_u$ using value iteration, that is, by computing $\bot\preccurlyeq \wpbellman_u \bot \preccurlyeq \cdots$ (\cref{subsec:fixedPtPrelim}). This is analogous to how Bellman--Ford's algorithm solves the shortest path problem in a weighted graph~\cite{bellman1958routing,ford1956network}. However, for shortest paths with non-negative weights, it is well-known that Bellman--Ford's algorithm can be improved to Dijkstra's algorithm~\cite{DBLP:journals/nm/Dijkstra59}.

We present a similar improvement for our current problem over widest path games. Note that our setting differs from classic Dijkstra's in two aspects: (i) the choice of weight semirings (i.e., shortest path vs. widest path, though this difference is not essential), and (ii) the presence of both Maximizer and Minimizer in our setting. The latter difference is nontrivial, and we address it by discarding Minimizer's suboptimal choices (\cref{algo_wpg:line9}).

\begin{mydefinition}[$\textsc{\textrm{DijkstraWidestPathGame}}(\sg, u)$]
Our Dijkstra-type algorithm for computing the widest path value function $W_{u}$ (\cref{def:widestPathValueFunc}) is in \cref{algo_wpg}.
\end{mydefinition}

\begin{algorithm}[t]
\caption{our Dijkstra-type algorithm for computing $W_{u}$ (\cref{def:widestPathValueFunc}).}\label{algo_wpg}
\DontPrintSemicolon
\SetAlgoNoLine
$\textsc{DijkstraWidestPathGame}(\sg, u)$\\ \Indp
    $\mathbb{W}(s) \gets 1$ for all $s \in T$; $\mathbb{W}(s) \gets 0$ for all $s \in S \setminus T$ \\
    \While{\textbf{\textup{true}}}{
        Find a state-action pair $(s,a) \in S \times A$ such that (i) $\mathbb{W}(s) = 0$, \linebreak (ii) $a \in \textsc{Av}(s)$, and (iii) $\min(\phi_u(s,a),\max_{s' \in \textsc{Post}(s,a)} \mathbb{W}(s'))$ is \linebreak the largest among those pairs which satisfy (i)--(ii) \label{algo_wpg:line5} \\
        \lIf{\textup{no such pair exist}}{\textbf{break}}
        $x \gets \min(\phi_u(s,a),\max_{s' \in \textsc{Post}(s,a)} \mathbb{W}(s'))$ \\
        \lIf{$x = 0$}{\textbf{break}}\label{algo_wpg:line8}
        \tikzmk{A}\label{algo_wpg:line9}\lIf{$s \in \MinState$ \textbf{\textup{and}} $|\textup{\textsc{Av}}(s)| > 1$}{$\textsc{Av}(s) \gets \textsc{Av}(s) \setminus \{a\}$}\tikzmk{B}\boxit{gray}
        \lElse{$\mathbb{W}(s) \gets x$}\label{algo_wpg:line10}
    }
    \Return $\mathbb{W}$
\end{algorithm}

\begin{myexample}\label{ex:algo1run}
Consider an execution of \cref{algo_wpg} with WPG constructed from SG in \cref{fig:ex_sg} (left) and $u : s_i \mapsto \text{\sfrac{$i$}{5}}$ for $0 \leq i \leq 5$. The width $\phi_u(s,a)$ of each state-action pair is shown in \cref{fig:cwpvf} (right). The algorithm begins with $\mathbb{W}(s_2) \gets 1$ and $\mathbb{W}(s_i) \gets 0$ for $i \neq 2$. In the while loop, the first pair selected is $(s_5,\alpha)$ with $x = \text{\sfrac{3}{5}}$. As $s_5 \in \MinState$ and $|\textsc{Av}(s_5)| > 1$, the action $\alpha$ is considered suboptimal and is removed from $s_5$. The next pairs are $(s_1,\gamma)$ and $(s_1,\alpha)$. Both are removed for the same reason. Then, $(s_0,\alpha)$ with $x = \text{\sfrac{1}{5}}$ is selected. Since $s_0 \in \MaxState$, the algorithm assigns $\mathbb{W}(s_0) \gets \text{\sfrac{1}{5}}$. The loop continues with $(s_1,\beta)$ selected. As now $\textsc{Av}(s_1) = \{\beta\}$, the algorithm performs $\mathbb{W}(s_1) \gets \text{\sfrac{3}{20}}$. The next pair is $(s_3,\alpha)$ with $x = 0$. Therefore, the loop ends and the algorithm returns $\mathbb{W} = (\text{\sfrac{1}{5}},\text{\sfrac{3}{20}},1,0,0,0)$.
\end{myexample}

\begin{mytheorem}\label{thm:computeWidestPathValueFunctionCorrectness}
$\textup{\textsc{DijkstraWidestPathGame}}(\sg, u)$ always terminates and returns $W_u$ for all inputs $(\sg,u)$.
\end{mytheorem}
\begin{proof}
\cref{algo_wpg} always terminates since each pair $(s,a)$ is considered at most once, and there are only finitely many such pairs. For its correctness, we show that, for each $s \in S$, when a positive value is assigned to $\mathbb{W}(s)$, that value is $W_u(s)$. The algorithm begins with $\mathbb{W}(s) \gets 1$ for each $s \in T$, which is correct as $W_u(s) = 1$. We then argue that \cref{algo_wpg:line10} assigns the correct value for each $s \in S \setminus T$.

First, we show that when a pair $(s,a)$ is selected in \cref{algo_wpg:line5}, the value $x = \min(\phi_u(s,a),\max_{s' \in \textsc{Post}(s,a)} \mathbb{W}(s'))$ is the maximum widest path width that can be achieved at $s$. For contradiction, suppose $x$ is not the maximum. Then, there must be a finite path $\rho = s_0 a_0 s_1 a_1 \ldots s_k$ where $s_0 = s$, $s_k \in T$, and $w_u(\rho) > x$. Along this path, there must be some $s_i$ such that $\mathbb{W}(s_i) = 0$ and $\mathbb{W}(s_{i+1}) \neq 0$. As $\mathbb{W}(s_{i+1}) \neq 0$, its value must have been assigned and, by our invariant, equals $W_{u}(s_{i+1})$. Since $w_u(\rho) > x$, it follows that $\phi_u(s_i,a_i) > x$ and $\mathbb{W}(s_{i+1}) > x$. Hence, $\min(\phi_u(s_i,a_i),\max_{s' \in \textsc{Post}(s_i,a_i)} \mathbb{W}(s')) > x$, contradicting the assumption that the selected pair $(s,a)$ in \cref{algo_wpg:line5} has the largest such value.

Now, we consider two cases. If $s\in \MaxState$, we set $\mathbb{W}(s) \gets x$ in \cref{algo_wpg:line10}, and by the above argument, this equals $W_u(s)$. If $s\in \MinState$, we remove $a$ from $\textsc{Av}(s)$ (unless $a$ is the last one), which is fine since doing so eliminates a suboptimal action for Minimizer. If the loop terminates with $x = 0$ at \cref{algo_wpg:line8}, then the maximum widest path width is zero for all remaining states. Thus, $\mathbb{W}(s) = 0 = W_u(s)$. \qed
\end{proof}

\section{Our 2WP-BVI Algorithm}\label{sec:WPGBVI}

We introduce a new operator built upon the widest path value function (\cref{def:widestPathValueFunc}).

\begin{mydefinition}[widest path operator $\wpg$]\label{def:WPOpr}
Let $\sg$ be an SG as in \cref{def:SG}. Its \emph{widest path operator} is $\wpg : [0,1]^S \to [0,1]^S$ where $\wpg u := W_u$ for $u \in [0,1]^S$.
\end{mydefinition}

A similar name appears for 1WP-BVI in~\cite{DBLP:conf/cav/PhalakarnTHH20}, but it refers to a different construction. The widest path operator $\wpg$ for our 2WP-BVI algorithm solves a \emph{(2-player)} widest path game w.r.t. a function $u$. As discussed earlier, the structure of (2-player) widest path games closely mirrors that of SGs. In contrast, the operator for 1WP-BVI performs \emph{(1-player, Maximizer-only)} widest path computation on a certain weighted directed graph. See \cref{subsec:compare} for a formal description of 1WP-BVI and a detailed comparison between the two approaches.

\begin{mydefinition}[$\textsc{2WP-BVI}(\sg, \varepsilon)$]\label{def:ourBVI}
Our 2WP-BVI algorithm for SG optimal reachability probabilities is in \cref{algo_iiwpg}.
\end{mydefinition}

\begin{algorithm}[t]
\caption{our 2WP-BVI algorithm, using widest path games for over-approximation. The input $\varepsilon > 0$ is a desired precision.}\label{algo_iiwpg}
\DontPrintSemicolon
\SetAlgoNoLine
$\textsc{2WP-BVI}(\sg, \varepsilon)$\\ \Indp
    $\ell(s) \gets 0$ for all $s \in S$; $u(s) \gets 1$ for all $s \in S$ \\
    \While{$u(\init) - \ell(\init) > 2\varepsilon$}{
        $\ell \gets \bellman \ell$\tcp*[f]{\cref{def:reachabilityBellmanOpr}}\\
        \tikzmk{A}$u \gets \wpg u$\tikzmk{B}\boxitt{gray}\tcp*[f]{\cref{def:WPOpr}}}
    \Return $\frac{1}{2}(\ell(\init) + u(\init))$
\end{algorithm}

\begin{myexample}\label{ex:algo2run}
We apply 2WP-BVI to \cref{fig:ex_sg} (left). The under-approximation sequence $\ell$ is the same as in \cref{ex:bvi}. The over-approximation sequence $u$, however, differs from \cref{ex:bvi} as we apply $\wpg$ instead of $\bellman$. Explicitly, our $u$ is $\top = (1,1,1,1,1,1) \succcurlyeq (1,1,1,0,0,0) \succcurlyeq (1,\text{\sfrac{3}{4}},1,0,0,0) \succcurlyeq (\text{\sfrac{11}{12}},\text{\sfrac{3}{4}},1,0,0,0) \succcurlyeq (\text{\sfrac{8}{9}},\text{\sfrac{11}{16}},1,0,0,0) \succcurlyeq \cdots$ converging to $\nu\wpg = V = (\text{\sfrac{4}{5}},\text{\sfrac{3}{5}},1,0,0,0)$ (by \cref{thm:key}).
\end{myexample}

\subsection{Correctness and Convergence}\label{subsec:ccAlg2}

For our 2WP-BVI algorithm (\cref{algo_iiwpg}), we prove its correctness (that $\ell$ and $u$ are under- and over-approximations, respectively) and convergence (of $\ell$ and $u$ to $V$). The latter also implies termination.

\begin{mylemma}\label{lem:VisFixWPOpr}
(i) $\wpg$ is monotone and $\omega^{\op}$-cocontinuous (see \cref{thm:kleene}), and (ii) $V = \wpg V = W_V$.
\end{mylemma}
\begin{proof}[of (ii)]
We first show $W_V \preccurlyeq V$ by contradiction. Assume there is $s \in S$ such that $W_V(s) > V(s) \geq 0$. It is not possible that $s \in T$ as $W_V(s) = V(s) = 1$, thus $s \in S \setminus T$. By $W_V(s) > 0$, there is a finite path $\rho' = s_0 a_0 s_1 a_1 \dotsc s_k$ where $s_0 = s$, $s_k \in T$, and $W_V(s) = w_V(\rho')$. By \cref{def:pathWidth}, $w_V(\rho') = \min_{0 \leq i < k} \phi_V(s_i,a_i)$ and we have $w_V(\rho') \leq \phi_V(s,a_0)$ in particular. In case of $s \in \MaxState \setminus T$, this leads to a contradiction $\phi_V(s,a_0) \geq W_V(s) > V(s) = \max_{a \in \textsc{Av}(s)} \phi_V(s,a)$, where the last equality is from \cref{def:reachabilityBellmanOpr}. In case of $s \in \MinState \setminus T$, we know that $\rho'$ gives the minimum widest path width at $s$. However, $\phi_V(s,a_0) \geq W_V(s) > V(s) = \min_{a \in \textsc{Av}(s)} \phi_V(s,a)$. Therefore, taking an action $\arg\min_{a \in \textsc{Av}(s)} \phi_V(s,a)$ would lead to a path with a smaller widest path width (its bottleneck is at most $\min_{a \in \textsc{Av}(s)} \phi_V(s,a)$), contradicting the optimality of $\rho'$.

We prove $V = \mu\bellman \preccurlyeq \mu\wpbellman_V = W_V$ (cf.\ \cref{lem:wpvf}) via induction that $\bellman^i \bot \preccurlyeq (\wpbellman_V)^i \bot$ for $i \in \mathbb{N}$. It is trivial for $i = 0$. For the induction $i+1$, we show $\bellman(\bellman^i \bot) \preccurlyeq \wpbellman_V((\wpbellman_V)^i\bot)$ by cases. It holds when $s \in T$. When $s \in S \setminus T$, we have $\bellman^i \bot \preccurlyeq V$, implying $\phi_{\bellman^i \bot}(s,a) \leq \phi_V(s,a)$ for $a \in \textsc{Av}(s)$. Also, $\phi_{\bellman^i \bot}(s,a) \leq \max_{s' \in \textsc{Post}(s,a)} (\bellman^i \bot)(s') \leq \max_{s' \in \textsc{Post}(s,a)} ((\wpbellman_V)^i \bot)(s')$ by our hypothesis. So, $\phi_{\bellman^i \bot}(s,a) \leq \min(\phi_V(s,a), \max_{s' \in \textsc{Post}(s,a)} ((\wpbellman_V)^i \bot)(s'))$. As LHS is $(\bellman(\bellman^i))\bot$ and RHS is $(\wpbellman_V(\wpbellman_V)^i)\bot$ (\cref{def:reachabilityBellmanOpr,def:wpbellman}), the proof concludes. \qed
\end{proof}

The following is our key theorem regarding a unique fixed point of the widest path operator $\wpg$. It exploits the maximality inheritance principle (\cref{sec:maxPresPrin}); observe its similarity to the proof of \cref{prop:maxPrevMC}, especially towards its end.
\begin{mytheorem}[key theorem: unique fixed point of $\wpg$]\label{thm:key}
In the setting of\linebreak\cref{def:WPOpr}, we have $V = \mu\wpg = \nu\wpg$. In particular, $\wpg$ has a unique fixed point.
\end{mytheorem}
\begin{proof}
We first show that $V = \mu\wpg$. Because $V = \wpg V$ (\cref{lem:VisFixWPOpr}), it suffices to prove that $f = \wpg f$ implies $f \succcurlyeq V$. Suppose $f = \wpg f$, which yields $f =\wpg f = W_f = \mu\wpbellman_f$, and thus $f = \wpbellman_f f$. It is also easy to see from \cref{lem:maxVsAvg} that $\phi_f(s,a) \leq \max_{s' \in \textsc{Post}(s,a)} f(s')$ for all $s \in S \setminus T$ and $a \in \textsc{Av}(s)$. Thus, the term $\min(\phi_f(s,a), \max_{s' \in \textup{\textsc{Post}}(s,a)} f(s'))$ in \cref{def:wpbellman} for $\wpbellman_f f$ becomes $\phi_f(s,a)$, matching exactly with \cref{def:reachabilityBellmanOpr} for $\bellman f$. This implies $\wpbellman_f f = \bellman f$, and so $f = \bellman f$. Because $V = \mu\bellman$ and $f$ is a fixed point of $\bellman$, we get $f \succcurlyeq V$.

It follows that $V \preccurlyeq \nu\wpg$. We shall show that $V = \nu\wpg$. We argue by contradiction; assume $V \prec \nu\wpg$. Then, there is $s \in S$ with $V(s) < \nu\wpg(s)$. Let $s^\star$ be a \emph{gap maximizer}, i.e., $(\nu\wpg - V)(s^\star) \geq (\nu\wpg - V)(s)$ for each $s \in S$. Note that we can pick such $s^\star$ as $S$ is finite. By our choice of $s^\star$, we get $(\nu\wpg-V)(s^\star) > 0$. It follows that $s^\star \not\in T$ (otherwise $V(s^\star) = 1$ and the gap is $\le 0$).

We now distinguish cases on whether $s^\star \in \MaxState$ or $s^\star \in \MinState$, and show that, in both cases, maximality is inherited under suitable actions.

Assume $s^\star \in \MaxState$. Since $\nu\wpg = \wpg(\nu\wpg) = W_{\nu\wpg} = \mu\wpbellman_{\nu\wpg}$, we get $\nu\wpg = \wpbellman_{\nu\wpg}(\nu\wpg)$ and there must be an action $a_\wpbellman \in \textsc{Av}(s^\star)$ such that $a_\wpbellman$ is optimal for the widest path objective (i.e., $\nu\wpg(s^\star) = \min(\phi_{\nu\wpg}(s^\star,a_\wpbellman), \max_{s' \in \textsc{Post}(s^\star,a_\wpbellman)} \nu\wpg(s'))$). By \cref{lem:maxVsAvg}, $\nu\wpg(s^\star) = \phi_{\nu\wpg}(s^\star,a_\wpbellman)$. This action $a_\wpbellman$, however, can be suboptimal for the reachability objective, giving $V(s^\star) \geq \phi_V(s^\star,a_\wpbellman)$. Thus, $(\nu\wpg - V)(s^\star) \leq \phi_{\nu\wpg}(s^\star,a_\wpbellman) - \phi_V(s^\star,a_\wpbellman) = \phi_{(\nu\wpg-V)}(s^\star,a_\wpbellman) \leq \max_{s' \in \textsc{Post}(s^\star,a_\wpbellman)} (\nu\wpg - V)(s') \leq (\nu\wpg - V)(s^\star)$. Hence, by \cref{lem:maxVsAvg},
$(\nu\wpg - V)(s^\star) = (\nu\wpg - V)(s')$ for all $s' \in \textsc{Post}(s^\star,a_\wpbellman)$. Therefore, being a gap maximizer is inherited under $a_\wpbellman$.

Now assume $s^\star \in \MinState$. Since $V = \mu\bellman$, we get $V = \bellman V$ and there must be an action $a_\bellman \in \textsc{Av}(s^\star)$ such that $a_\bellman$ is optimal for the reachability objective (i.e., $V(s^\star) = \phi_V(s^\star,a_\bellman)$). This action, however, can be suboptimal for the widest path objective, giving $\nu\wpg(s^\star) \leq \min(\phi_{\nu\wpg}(s^\star,a_\bellman), \max_{s' \in \textsc{Post}(s^\star,a_\bellman)} \nu\wpg(s')) = \phi_{\nu\wpg}(s^\star,a_\bellman)$, where the equality is by \cref{lem:maxVsAvg}(i). So, $(\nu\wpg - V)(s^\star) \leq \phi_{\nu\wpg}(s^\star,a_\bellman) - \phi_V(s^\star,a_\bellman) \leq (\nu\wpg - V)(s^\star)$. Thus, by \cref{lem:maxVsAvg}(ii), $(\nu\wpg - V)(s^\star) = (\nu\wpg - V)(s')$ for all $s' \in \textsc{Post}(s^\star,a_\bellman)$. Hence, being a gap maximizer is inherited under $a_\bellman$.

We now look at the stochastic game in which Maximizer always chooses $\MaxStrategy(\cdot) = a_\wpbellman$ and Minimizer always chooses $\MinStrategy(\cdot) = a_\bellman$. The above reasoning is valid for any gap maximizer $s^\star$. So we conclude that all states $s$ reachable from $s^\star$ under $(\MaxStrategy,\MinStrategy)$ are gap maximizers (w.r.t.\ $\nu\wpg-V$). Since $\nu\wpg(s) = V(s) = 1$ for $s \in T$, it follows that $T$ is unreachable from $s^\star$ under $(\MaxStrategy,\MinStrategy)$.

Now note that, under the widest path objective and $(\MaxStrategy,\MinStrategy)$, Maximizer plays optimally while Minimizer may not. Thus, the widest path width at a state $s$ is at least $\nu\wpg(s)$. This is true for $s^\star$. Since $(\nu\wpg - V)(s^\star) > 0$, it implies $\nu\wpg(s^\star) > 0$. Therefore, under $(\MaxStrategy,\MinStrategy)$ and by the definition of the widest path width, there is a path from $s^\star$ to $T$ of width at least $\nu\wpg(s^\star) > 0$. This contradicts the above that $T$ is unreachable from $s^\star$ under $(\MaxStrategy,\MinStrategy)$. \qed
\end{proof}

\begin{mytheorem}\label{thm:ourBVICorrectness}
Our proposed algorithm $\textup{\textsc{2WP-BVI}}(\sg, \varepsilon)$ surely terminates and returns a value that is $\varepsilon$-close to $V(\init)$ for all inputs $(\sg,\varepsilon)$.
\end{mytheorem}
\begin{proof}
We get $\bot \preccurlyeq \bellman\bot \preccurlyeq \cdots \preccurlyeq \mu\bellman = V = \nu\wpg \preccurlyeq \cdots \preccurlyeq \wpg\top \preccurlyeq \top$ from the properties of $\bellman$ and $\wpg$ (\cref{thm:wellknown,thm:key}). Moreover, both under- and over-approximation sequences converge after $\omega$ steps, because the continuity of $\bellman$ (well-known) and $\wpg$ (\cref{lem:VisFixWPOpr}). Thus, $(\wpg^i\top)(\init)-(\bellman^i\bot)(\init) \leq 2\varepsilon$ holds for some $i \in \mathbb{N}$ and $\frac{1}{2}((\bellman^i\bot)(\init)+(\wpg^i\top)(\init))$ is $\varepsilon$-close to $V(\init)$. \qed
\end{proof}

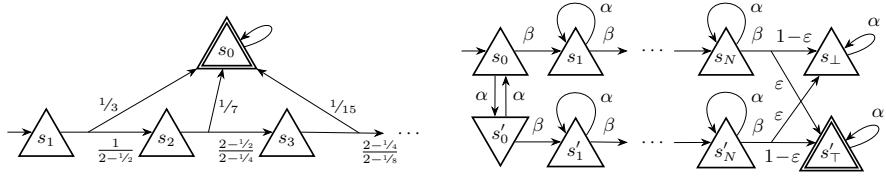
\begin{figure}[t]
    \centering
    \scalebox{.8}{\begin{tikzpicture}[
    maximizer/.style={regular polygon,regular polygon sides=3, inner sep=1pt, anchor=north},
    minimizer/.style={maximizer, shape border rotate=180}
]
    \node [state, maximizer] at (2,0) (s2) {$s_1$};
    \node [state, maximizer] at (4,0) (s3) {$s_2$};
    \node [state, maximizer] at (6,0) (s4) {$s_3$};
    \node [draw=none, anchor=north] at (8,-0.25) (s5) {$\ \cdots$};
    \node [state, maximizer, accepting] at (5,1.5) (s0) {$s_0$};

    \path ($(s2.150) - (3ex,0)$) edge (s2.150);
    \path (s0) edge[out=15, in=40, looseness=15] (s0);
    \path (s2.30) edge node[below, pos=0.6]{$\frac{1}{2-\textrm{\sfrac{1}{2}}}$} coordinate[pos=0.3](s23) (s3.150);
    \path (s23) edge node[above left=-3pt, pos=0.3]{\sfrac{1}{3}} (s0.210);
    \path (s3.30) edge node[below, pos=0.6]{$\frac{2-\textrm{\sfrac{1}{2}}}{2-\textrm{\sfrac{1}{4}}}$} coordinate[pos=0.3](s34) (s4.150);
    \path (s34) edge node[above right, pos=0.2]{\sfrac{1}{7}} (s0);
    \path (s4.30) edge node[below, pos=0.95]{$\frac{2-\textrm{\sfrac{1}{4}}}{2-\textrm{\sfrac{1}{8}}}$} coordinate[pos=0.7](s45) (s5);
    \path (s45) edge node[above right=-4pt, pos=0.3]{\sfrac{1}{15}} (s0.330);
\end{tikzpicture}}\quad\scalebox{.8}{\begin{tikzpicture}[
    maximizer/.style={regular polygon,regular polygon sides=3, inner sep=1pt, anchor=north},
    minimizer/.style={maximizer, shape border rotate=180}
]
    \node [state, maximizer] at (0,0) (s0) {$s_0$};
    \node [state, maximizer] at (1.25,0) (s1) {$s_1$};
    \node [draw=none, anchor=north] at (2.5,-0.225) (sdot) {$\ \cdots$};
    \node [state, maximizer] at (3.75,0) (sn) {\makebox[\widthof{$s_0$}][c]{$s_N$}};
    \node [state, maximizer] at (5.5,0) (snn) {\makebox[\widthof{$s_0$}][c]{$s_\bot$}};
    \node [state, minimizer] at (0,-1.5) (ss0) {$s'_0$};
    \node [state, maximizer] at (1.25,-1.5) (ss1) {$s'_1$};
    \node [draw=none, anchor=north] at (2.5,-1.75) (ssdot) {$\ \cdots$};
    \node [state, maximizer] at (3.75,-1.5) (ssn) {\makebox[\widthof{$s'_0$}][c]{$s'_N$}};
    \node [state, maximizer, accepting] at (5.5,-1.5) (ssnn) {\makebox[\widthof{$s'_0$}][c]{$s'_\top$}};

    \path (s0.252) edge node[left]{$\alpha$} (ss0.108);
    \path (ss0.72) edge node[right]{$\alpha$} (s0.288);

    \path ($(s0.150) - (3ex,0)$) edge (s0.150);
    \path (s0.30) edge node[above, pos=0.3]{$\beta$} (s1.150);
    \path (s1) edge[out=55, in=120, looseness=8] node[right=2pt, pos=0.35]{$\alpha$} (s1);
    \path (s1.30) edge node[above]{$\beta$} (sdot);
    \path (sdot) edge (sn.150);
    \path (sn) edge[out=55, in=120, looseness=8] node[right=2pt, pos=0.35]{$\alpha$} (sn);
    \path (sn.30) edge node[above, pos=0.225]{$\beta$} coordinate[pos=0.4](sh) node[above, pos=0.7]{$1\!-\!\varepsilon$} (snn.150);
    \path (sh) edge node[left, pos=0.4]{$\varepsilon$} (ssnn);
    \path (snn) edge[out=15, in=40, looseness=15] node[above=2pt]{$\alpha$} (snn);

    \path (ss0.330) edge node[above]{$\beta$} (ss1.150);
    \path (ss1) edge[out=55, in=120, looseness=8] node[right=2pt, pos=0.35]{$\alpha$} (ss1);
    \path (ss1.30) edge node[above]{$\beta$} (ssdot);
    \path (ssdot) edge (ssn.150);
    \path (ssn) edge[out=55, in=120, looseness=8] node[right=2pt, pos=0.35]{$\alpha$} (ssn);
    \path (ssn.30) edge node[above, pos=0.225]{$\beta$} coordinate[pos=0.4](ssh) node[below, pos=0.55]{$1\!-\!\varepsilon$} (ssnn.147);
    \path (ssh) edge node[left, pos=0.4]{$\varepsilon$} (snn);
    \path (ssnn) edge[out=15, in=40, looseness=15] node[above=2pt]{$\alpha$} (ssnn);
\end{tikzpicture}}
    \caption{Left: an infinite-state SG giving $\mu\bellman \prec \nu\bellman$ and $\mu\wpg \prec \nu\wpg$. Right: the benchmark \emph{ECchain} ($N$ is the model parameter). The labels $1$ for $\delta(s,a,s') = 1$ are omitted.}
    \label{fig:ex_infsg}\label{fig:chain}
\end{figure}

\begin{myremark}\label{remark:finiteness}
The maximality inheritance principle (cf.\ \cref{sec:maxPresPrin}) requires the state space to be finite. Consider \cref{fig:ex_infsg} (left) as an example. For $i > 0$, the chance to go from $s_i$ to $s_{i+1}$ is $\frac{2-(\textrm{\sfrac{1}{2}})^{i-1}}{2-(\textrm{\sfrac{1}{2}})^i}$ while to $s_0$ is $\frac{1}{2^{i+1}-1}$. This results in $\mu\bellman = \mu\wpg = V$, where $V(s_i) = 1 - \frac{2-(\textrm{\sfrac{1}{2}})^{i-1}}{2} = \frac{1}{2^i}$ for $i \in \mathbb{N}$, whereas $V \prec \nu\bellman = \nu\wpg = \top$.
\end{myremark}

\subsection{1WP-BVI and Its Comparison with 2WP-BVI}\label{subsec:compare}

\begin{algorithm}[t]
\caption{1WP-BVI algorithm of~\cite{DBLP:conf/cav/PhalakarnTHH20}.}\label{algo_cav20}
\DontPrintSemicolon
\SetAlgoNoLine
$\textsc{1WP-BVI}(\sg, \varepsilon)$\\ \Indp
    $\ell(s) \gets 0$ for all $s \in S$; $u(s) \gets 1$ for all $s \in S$ \\
    \While{$u(\init) - \ell(\init) > 2\varepsilon$}{
        $\ell \gets \bellman \ell$ \\
        \tikzmk{A}$\sg' \gets \textsc{PlyrRdct}(\sg,\ell)$; $u \gets \wpg_{\sg'} u$\tikzmk{B}\boxitt{gray}}
    \Return $\frac{1}{2}(\ell(\init) + u(\init))$
\end{algorithm}

Recall that our 2WP-BVI improves upon \emph{1WP-BVI} of~\cite{DBLP:conf/cav/PhalakarnTHH20}. Here, we formally describe 1WP-BVI and provide a comparison between 1WP-BVI and 2WP-BVI.

1WP-BVI is shown as \cref{algo_cav20} and works as follows: the under-approximation sequence is computed using the Bellman operator $\bellman$ (\cref{def:reachabilityBellmanOpr}). For the over-approximation sequence, it starts with \emph{player reduction} $\sg' \gets \textsc{PlyrRdct}(\sg,\ell)$, which removes Minimizer's suboptimal actions w.r.t. $\ell$ and converts all Minimizer's states into Maximizer's. Formally, this yields an MDP $\sg'$ where (i) for $s \in \MinState$, $\textsc{Av}'(s) = \arg\min_{a \in \textsc{Av}(s)} \phi_\ell(s,a)$, and (ii) $\MaxStateX = S$, $\MinStateX = \emptyset$.

Then, 1WP-BVI computes a (1-player, Maximizer-only) widest path width on a certain graph derived from the resulting MDP. In \cref{algo_cav20}, this is $u \gets \wpg_{\sg'}u$. In this work's notation, it corresponds to $u \gets \textsc{DijkstraWidestPathGame}(\sg',u)$ (\cref{algo_wpg}). Since $\MinStateX = \emptyset$, \cref{algo_wpg:line9} of \cref{algo_wpg} is skipped. So, 1WP-BVI can outperform 2WP-BVI, as removing Minimizer's suboptimal actions in 2WP-BVI is costly.

However, 1WP-BVI has two drawbacks. Firstly, player reduction complicates the convergence proof for 1WP-BVI. Since the MDP $\sg'$ can vary in each iteration,~\cite[Thm. 4.11]{DBLP:conf/cav/PhalakarnTHH20} provides a non-constructive proof via the infinitary pigeonhole principle. In contrast, the fixed-point uniqueness proof for 2WP-BVI (\cref{thm:key}) is cleaner and more structured, potentially enabling further extensions.

Secondly, player reduction in 1WP-BVI imposes a dependency between the under- and over-approximations. For example, slow convergence of the under-approximations delays that of the over-approximations.

As an instance, this is seen in the SGs of the \emph{ECchain} benchmarks (\cref{fig:chain}, right). In these SGs, the top path has reachability $\varepsilon$ and the bottom path $1-\varepsilon$. In the execution of 1WP-BVI (\cref{algo_cav20}), it takes roughly $N$ iterations for $\ell(s'_\top) = 1$ to propagate to $s'_1$. Thus, player reduction is ineffective: no actions of $s'_0$ are removed, and $u(s_0)$ remains $1-\varepsilon$ for about $N$ iterations, until $\ell(s'_1) \neq 0$. In contrast, 2WP-BVI detects the suboptimal action $\beta$ for $s'_0$ in the first iteration, and can assign $u(s_0) = \varepsilon$. Since $\ell(s_0) = 0$, 2WP-BVI immediately terminates.

\section{Related Work}\label{sec:relatedWork}

Besides the works discussed elsewhere, we review the following closely related works. To our best knowledge, these are the only available VI-based tools for SG reachability that provide precision guarantees. They are compared with 1WP-BVI and our 2WP-BVI via experiments in \cref{sec:exp}.

\myparagraphii{End component-based BVI (EC-BVI)}{\cite{DBLP:conf/cav/KelmendiKKW18}.} This is the first work on BVI for SG reachability. In each iteration, it first performs \emph{player reduction} to obtain an MDP. Then, end-component analysis~\cite{DBLP:conf/atva/BrazdilCCFKKPU14,DBLP:journals/tcs/HaddadM18}---called \emph{deflating}---is applied to the MDP to produce a tighter over-approximation. This work was later extended to other quantitative objectives such as total reward and mean payoff~\cite{DBLP:conf/lics/KretinskyMW23}.

\myparagraphii{PET2\footnote[3]{On the name PET~\cite{qcomp} vs. PET2~\cite{DBLP:conf/cav/MeggendorferW24}: in~\cite{qcomp}, the tool from the PET2 team is called PET (``PET has recently been extended to support reachability objectives for SGs.''). Meanwhile,~\cite{DBLP:conf/cav/MeggendorferW24} says ``PET2 is an extension of PET1~\cite{DBLP:conf/atva/Meggendorfer22} (only applicable to MDPs).'' Given these, we assume that the discussions on PET in~\cite{qcomp} pertain to PET2 from~\cite{DBLP:conf/cav/MeggendorferW24}. This aligns with our understanding of PET2 based on~\cite{DBLP:conf/cav/MeggendorferW24} and its source code.}}{\cite{DBLP:conf/cav/MeggendorferW24}.} It extends PET~\cite{DBLP:conf/atva/Meggendorfer22} and is based on EC-BVI. It incorporates partial exploration techniques and various performance optimizations.

\myparagraphii{Optimistic value iteration (OVI)}{\cite{DBLP:conf/atva/AzeemEKSW22}.} It extends OVI for MDPs~\cite{DBLP:conf/cav/QuatmannK18} to SGs. Given $\ell \preccurlyeq \mu\bellman$, it checks if $\mu\bellman \preccurlyeq u := \ell + 2\varepsilon$. If so, $\ell+\varepsilon$ is $\varepsilon$-close to $\mu\bellman$. This can be done via Knaster--Tarski's theorem~\cite{knaster,tarski} or Park's induction~\cite{park}: $\bellman u \preccurlyeq u$ implies $\mu\bellman \preccurlyeq u$. End-component analysis is needed to ensure termination.

\myparagraphii{Sound value iteration (SVI)}{\cite{DBLP:journals/corr/abs-2411-11549}.} It extends SVI for MDPs~\cite{DBLP:conf/cav/HartmannsK20} to SGs. Using the set $S_{=0}$ (cf. \cref{sec:maxPresPrin}), the reachability from $s$ to $T$ is over-approximated by that of reaching $T$ from $s$ in $k$ steps, plus that of remaining in $S \setminus (T \cup S_{=0})$ for $k$ steps. It employs end-component analysis and other iterative mechanisms. According to~\cite{DBLP:journals/corr/abs-2411-11549}, SVI can take fewer iterations to terminate compared to EC-BVI.

\section{Experiments}\label{sec:exp}

We implemented our proposed 2WP-BVI (\cref{algo_iiwpg}) as an extension of PRISM-games~\cite{DBLP:conf/cav/KwiatkowskaN0S20}. Our artifact is available at~\cite{artifact}. The procedure $u \gets \wpg u$ in \cref{algo_iiwpg} is concretely implemented as $u \gets \textup{\textsc{DijkstraWidestPathGame}}(\sg, u)$ (\cref{algo_wpg}).

Experiments were performed on MacBook Pro\textsuperscript{\textregistered} with 2.3 GHz Intel\textsuperscript{\textregistered} Core\textsuperscript{\texttrademark} i9 and 32 GB of RAM. Each tool was executed inside a Docker container with OpenJDK JRE-17. The time-out (T/O) is 10 minutes. The precision is $\varepsilon = 10^{-6}$.

\begin{table}[p]
    \centering
    \caption{Experimental results. Following each benchmark name is a model parameter (the bigger the more complex). The execution times are in seconds; T/O (time-out) is 600 seconds. The fastest algorithms for each benchmark are shown in bold.}\label{tab:exp}
    \newlength{\maxlen}
\settowidth{\maxlen}{\textbf{2WP-BVI\,}}
\newcommand*{\head}[1]{%
    \begin{sideways}
      \makebox[\maxlen][l]{\textbf{#1}}
    \end{sideways}}

\begin{tabular}{lC{0.75cm}R{1cm}R{1cm}R{1cm}R{1cm}R{1cm}R{1cm}}
\toprule
\multicolumn{1}{c}{\raisebox{1pt}{\textbf{Benchmarks}}} & \head{Sources} & \head{EC-BVI} & \head{PET2} & \head{OVI} & \head{SVI} & \head{1WP-BVI} & \head{2WP-BVI} \head{(ours)} \\
\toprule
avoid (10\texttimes 10, exit) & \cite{DBLP:conf/cav/ChatterjeeKWW20} & 21.4 & \textbf{7.6} & 24.7 & T/O & 8.5 & 8.5 \\
avoid (10\texttimes 10, find) & & 13.4 & \textbf{5.0} & 17.6 & 13.8 & 15.5 & 17.2 \\
avoid (15\texttimes 15, exit) & & 185.4 & 28.0 & 191.6 & T/O & \textbf{26.8} & 32.2 \\
avoid (15\texttimes 15, find) & & 208.4 & \textbf{19.1} & 127.3 & 117.4 & 63.2 & 67.0 \\
avoid (20\texttimes 20, exit) & & T/O & 104.7 & T/O & T/O & \textbf{90.3} & 110.4 \\
avoid (20\texttimes 20, find) & & 551.1 & \textbf{68.6} & T/O & 564.3 & 178.5 & 188.0 \\
\midrule
cloud (5) & \cite{DBLP:conf/monterey/CalinescuKJ12} & 293.8 & T/O & 2.8 & 2.8 & \textbf{1.9} & \textbf{1.9} \\
cloud (6) & & T/O & T/O & 5.8 & 5.4 & \textbf{3.2} & 3.7 \\
cloud (7) & & T/O & T/O & 19.2 & 20.2 & \textbf{6.9} & 8.1 \\
\midrule
dice (25) & \cite{DBLP:conf/cav/KwiatkowskaN0S20} & 2.9 & \textbf{2.7} & 2.9 & 3.7 & 2.8 & 3.0 \\
dice (50) & & 7.6 & \textbf{4.3} & 8.2 & 8.8 & 6.9 & 8.8 \\
dice (100) & & 45.3 & \textbf{11.3} & 49.9 & 51.5 & 20.2 & 24.5 \\
\midrule
hallway human (5\texttimes 5) & \cite{DBLP:conf/cav/ChatterjeeKWW20} & 2.5 & 2.7 & \textbf{2.4} & 2.5 & 15.8 & 13.4 \\
hallway human (10\texttimes 10) & & 11.7 & \textbf{10.2} & 11.7 & 11.9 & T/O & T/O \\
hallway human (15\texttimes 15) & & 78.0 & \textbf{48.7} & 59.3 & 59.3 & T/O & T/O \\
\midrule
investor (50) & \cite{DBLP:journals/tocl/McIverM07} & 21.8 & \textbf{11.5} & 27.1 & 38.5 & 25.2 & 26.3 \\
investor (100) & & 155.2 & \textbf{78.0} & 189.6 & 228.1 & 138.9 & 141.9 \\
\midrule
investors (2, 10) & \cite{DBLP:journals/fmsd/ChenFKPS13} & 8.4 & \textbf{6.0} & 12.5 & 12.3 & 8.4 & 8.1 \\
investors (2, 20) & & 31.9 & \textbf{20.1} & 39.5 & 82.6 & 46.8 & 50.3 \\
investors (2, 40) & & 166.6 & \textbf{86.4} & 200.2 & 346.5 & 252.8 & 271.3 \\
\midrule
mdsm (3) & \cite{5952980} & 4.8 & \textbf{4.2} & 5.0 & 8.8 & 5.1 & 5.2 \\
mdsm (4) & & 20.2 & \textbf{12.1} & 19.1 & 31.2 & 14.1 & 17.5 \\
mdsm (5) & & 90.5 & 54.1 & 82.8 & 139.3 & \textbf{52.9} & 61.0 \\
\midrule
BigMec (5000) & \cite{DBLP:journals/iandc/KretinskyRSW22} & T/O & \textbf{6.9} & 9.9 & T/O & 8.4 & 9.2 \\
BigMec (7500) & & T/O & \textbf{8.7} & 19.7 & T/O & 17.2 & 19.8 \\
BigMec (10000) & & T/O & \textbf{18.7} & 33.0 & T/O & 28.6 & 35.3 \\
BigMec (25000) & & T/O & \textbf{90.2} & 252.5 & T/O & 196.7 & 267.1 \\
\midrule
ManyMecs (5000) & \cite{DBLP:journals/iandc/KretinskyRSW22} & T/O & \textbf{24.6} & T/O & T/O & 56.6 & 42.3 \\
ManyMecs (7500) & & T/O & \textbf{55.0} & T/O & T/O & 127.2 & 102.5 \\
ManyMecs (10000) & & T/O & \textbf{112.9} & T/O & T/O & 233.6 & 201.6 \\
ManyMecs (25000) & & T/O & T/O & T/O & T/O & T/O & T/O \\
\midrule
manyECs (5000) & \cite{DBLP:conf/cav/PhalakarnTHH20} & T/O & 27.2 & T/O & T/O & \textbf{7.3} & 7.7 \\
manyECs (7500) & & 57.6 & T/O & T/O & T/O & 17.4 & \textbf{16.8} \\
manyECs (10000) & & T/O & 108.6 & T/O & T/O & \textbf{30.5} & 30.6 \\
manyECs (25000) & & T/O & T/O & T/O & T/O & 221.8 & \textbf{220.7} \\
\midrule
ECchain (5000) & ours & T/O & T/O & 15.3 & 76.8 & 12.4 & \textbf{5.5} \\
ECchain (7500) & & T/O & T/O & 36.9 & 202.5 & 28.4 & \textbf{12.6} \\
ECchain (10000) & & T/O & T/O & 64.7 & 393.3 & 48.5 & \textbf{22.7} \\
ECchain (25000) & & T/O & T/O & 539.8 & T/O & 374.2 & \textbf{183.9} \\
\bottomrule
\end{tabular}
    \vspace{-33pt}
\end{table}

The experimental results are in \cref{tab:exp}. Most benchmarks are from existing works~\cite{DBLP:conf/cav/ChatterjeeKWW20,DBLP:conf/monterey/CalinescuKJ12,DBLP:conf/cav/KwiatkowskaN0S20,DBLP:journals/tocl/McIverM07,DBLP:journals/fmsd/ChenFKPS13,5952980,DBLP:journals/iandc/KretinskyRSW22,DBLP:conf/cav/PhalakarnTHH20}; several are real-world case studies. The last four are artificial: \emph{BigMec}, \emph{ManyMecs}, \emph{manyECs} feature large or numerous end components, while \emph{ECchain} (\cref{fig:chain}, right) are newly designed to make the under-approximation sequence converge slowly.

The results demonstrate notable performance advantages of WP-based BVIs (2WP-BVI and 1WP-BVI) in benchmarks with many end components---chiefly, \emph{manyECs} and \emph{ECchain}, where either 2WP-BVI or 1WP-BVI performs best. Note that \emph{ECchain} contains many self-loop end components (see \cref{fig:chain}, right).

The difference between 2WP-BVI and 1WP-BVI is signified in \emph{ECchain}, where 2WP-BVI is constantly faster than 1WP-BVI. This is because, in 1WP-BVI, the over-approximation depends on the under-approximation. In \emph{ECchain}, the latter converges slowly, which drags the convergence of the former. Conversely, 2WP-BVI is free from this performance burden (see \cref{subsec:compare} for details).

Regarding OVI's performance, it is comparable to 2WP-BVI and 1WP-BVI, except on SGs with many end components (\emph{ManyMecs}, \emph{manyECs}). EC-BVI and SVI are generally slower than the others across many benchmarks. That said, they perform quite well on certain cases (e.g., \emph{investors} for EC-BVI, \emph{cloud} for SVI). Thus, EC-BVI or SVI may be preferable depending on the benchmark.

We find that PET2 is the fastest on many benchmarks---this is consistent with the results from the QComp 2023 Competition~\cite{qcomp}. Still, we find that (i) even in those cases, the performance gap with 2WP-BVI is rarely substantial (an exception is \emph{hallway human}), and (ii) there are some benchmarks where 2WP-BVI clearly outperforms PET2 (namely, \emph{cloud}, \emph{manyECs}, and \emph{ECchain}).

Moreover, a discussion in~\cite{qcomp} attributes a large part of PET2's performance advantage to engineering efforts such as partial exploration, loop unrolling, and time-memory trade-offs~\cite{DBLP:conf/cav/MeggendorferW24}. This discussion, together with \cref{tab:exp}, points to the performance potential of 2WP-BVI. It would be interesting to observe its performance after similar engineering efforts have been applied.

Overall, the results demonstrate 2WP-BVI's potential for broad applicability: there are benchmarks for which ours is the best performer; for most of the other benchmarks, ours is at most several times slower than the best performer. We note that the current implementation of 2WP-BVI is a prototype and leaves room for optimization, as discussed above.

\section{Conclusions and Future Work}\label{sec:conclusion}

We presented 2WP-BVI, a novel BVI algorithm for SG reachability that is based on widest path computations---thereby avoiding the need for end-component analysis---and improves upon 1WP-BVI~\cite{DBLP:conf/cav/PhalakarnTHH20} in terms of theoretical clarity. 2WP-BVI is built on a construct called widest path games, and we introduced a proof principle referred to as the maximality inheritance principle. Our experimental evaluation demonstrated the algorithm's practical relevance.

We plan to optimize our prototype in various ways, such as incorporating partial exploration and applying low-level performance enhancements. Additionally, we aim to further investigate the theoretical foundations of our constructions through the lens of lattice theory.

\begin{credits}
\subsubsection{\ackname} The authors would like to thank the reviewers for their comments on improving the manuscript and acknowledge the implementations of Azeem et al.~\cite{DBLP:journals/corr/abs-2411-11549}, Kelmendi et al.~\cite{DBLP:conf/cav/KelmendiKKW18}, Meggendorfer~\cite{qcomp_impl}, Meggendorfer--Weininger~\cite{DBLP:conf/cav/MeggendorferW24}, and Phalakarn et al.~\cite{DBLP:conf/cav/PhalakarnTHH20}, which we use as basis for our implementation and experiments. This work is supported by ERATO HASUO Metamathematics for Systems Design Project (No. JPMJER1603) and the ASPIRE grant No. JPMJAP2301, JST.
\end{credits}

\appendix
\section*{Appendix: Omitted Proofs}\label{appendix:omittedProofs}

\theorembodyfont{\itshape}
\newtheorem*{lem32}[mytheorem]{Lemma 3.2}
\begin{lem32}[max.\ vs.\ average]
Let $S$ be a finite set, $f\colon S\to [0,1]$ be a function, and $\delta\in \dist(S)$ be a distribution over $S$. Assume that $s^\star$ is an $f$-maximizer (meaning $f(s^\star) \geq f(s)$ for each $s\in S$).
\begin{enumerate}[label=(\roman*)]
    \item $f(s^\star)\geq \sum_{s\in S}\delta(s)\cdot f(s)$ holds (``the maximum is above the average'').
    \item Assume further that we have an equality $f(s^\star) = \sum_{s\in S}\delta(s)\cdot f(s)$ (``the maximum coincides with the average''). Then, for each $s\in S$, $\delta(s)>0$ implies $f(s^\star) = f(s)$.
\end{enumerate}
\end{lem32}
\begin{proof}
For (i), it is clear that $\sum_{s\in S}\delta(s)\cdot f(s) \leq \sum_{s\in S}\delta(s)\cdot f(s^\star) = f(s^\star)$, as $f(s) \leq f(s^\star)$ for all $s \in S$. For (ii), we prove by contraposition. Let there be $s' \in S$ with $\delta(s') > 0$ and $f(s^\star) > f(s')$. Then, $\sum_{s\in S}\delta(s)\cdot f(s) = \sum_{s\in S \setminus \{s'\}}\delta(s)\cdot f(s) + \delta(s')\cdot f(s') \leq \sum_{s\in S \setminus \{s'\}}\delta(s)\cdot f(s^\star) + \delta(s')\cdot f(s') < f(s^\star)$. \qed
\end{proof}

\newtheorem*{lem44}[mytheorem]{Lemma 4.4}
\begin{lem44}
(i) $\wpbellman_u$ is monotone, and (ii) $\mu\wpbellman_u = W_u$ (\cref{def:widestPathValueFunc}).
\end{lem44}
\begin{proof}[of (ii)]
For $k \in \mathbb{N}$ and $s \in \MaxState$, the value $((\wpbellman_u)^{k+1}\bot)(s)$ is the maximum widest path width that can be achieved at $s$ over all strategies when looking at all finite paths of length at most $k$ (i.e., $s_0 a_0 s_1 a_1 \ldots s_k$) that start at $s$. For $s \in \MinState$, the same applies with $((\wpbellman_u)^{k+1}\bot)(s)$ being the minimum widest path width. Due to the fact that, if $W_u(s) > 0$ then there is a path of width $W_u(s)$ from $s$ to $T$ of length at most $|S|-1$ (i.e. with no loops), we know that $W_u = (\wpbellman_u)^{|S|}\bot = (\wpbellman_u)^{|S|+1}\bot$. By \cref{lem:wpbellmanMonotone} and \cref{thm:cousot}, we obtain the under-approximation sequence $\bot \preccurlyeq \wpbellman_u \bot \preccurlyeq \cdots \preccurlyeq (\wpbellman_u)^{|S|}\bot = (\wpbellman_u)^{|S|+1}\bot = \mu\wpbellman_u = W_u$. \qed
\end{proof}

\newtheorem*{lem54}[mytheorem]{Lemma 5.4}
\begin{lem54}
(i) $\wpg$ is monotone and $\omega^{\op}$-cocontinuous (see \cref{thm:kleene}), and (ii) $V = \wpg V = W_V$.
\end{lem54}
\begin{proof}[of (i)]
Monotonicity is straightforward. For $\omega^{\op}$-cocontinuity (that is, $\wpg(\inf_{n<\omega} u_{n})=\inf_{n<\omega}\wpg(u_n)$ for a decreasing chain $u_{0}\succcurlyeq u_{1}\succcurlyeq\cdots$), we notice the following: for any positive $\varepsilon$, there is $N \in \mathbb{N}$ that is large enough, such that (i) $u_{N}(s)<(\inf_{n<\omega} u_{n})(s)+\varepsilon$ for each $s \in S$, and thus (ii) $\phi_{u_{N}}(s,a)<\phi_{\inf_{n<\omega} u_{n}}(s,a)+\varepsilon$ for each $s \in S$ and $a \in \textsc{Av}(s)$. This holds because there are only finitely many $s$'s and $a$'s.

Moreover, when $\varepsilon$ is small enough---say, it is smaller than the smallest difference of widths under $\inf_{n<\omega} u_{n}$---we can show that $\wpg(u_{N}) = W_{u_{N}}$ is given by the same strategies $(\MaxStrategy,\MinStrategy)$ (cf. \cref{def:widestPathValueFunc}) and moreover by the same path $\rho$ (cf. \cref{eq:widestPath}). In this case, by \cref{def:pathWidth}, we have $\wpg(u_N)\preccurlyeq\wpg (\inf_{n<\omega} u_{n})+\varepsilon$. We can take such $N$ for any positive $\varepsilon$ that is small enough; this establishes $\wpg(\inf_{n<\omega} u_{n})=\inf_{n<\omega}\wpg(u_n)$. \qed
\end{proof}

\bibliographystyle{splncs04}
\bibliography{ref}

\end{document}